\documentclass[12pt,journal,letterpaper,onecolumn]{IEEEtran}

\usepackage{graphicx}
\usepackage{amsmath}
\usepackage{ifthen}
\usepackage{verbatim}
\usepackage{epsfig}
\usepackage[iso]{umlaute}
\usepackage{amssymb}
\usepackage{graphicx}
\usepackage{bbm}
\usepackage{subfigure}
\usepackage{ifthen}
\usepackage{url}
\usepackage{algorithm}
\usepackage{amsthm}
\usepackage{algorithm}
\usepackage{algorithmic}

\newcommand{\DB}{\ensuremath{\mathcal{D}}}
\newcommand{\RR}{\ensuremath{\mathbb{R}}}

\newcommand{\MinDist}{\ensuremath{\textit{MinDist}}} 
\newcommand{\MaxDist}{\ensuremath{\textit{MaxDist}}} 

\newtheorem{definition}{Definition}
\newtheorem{lemma}{Lemma}
\newtheorem{corollary}{Corollary}
\newtheorem{example}{Example}

\begin{document}

\newboolean{TR}
\setboolean{TR}{false}

\title{A Novel Probabilistic Pruning Approach to Speed Up Similarity Queries in Uncertain Databases}

\author{Thomas Bernecker,
        Tobias Emrich,
        Hans-Peter Kriegel,
        Nikos Mamoulis,
        Matthias Renz,
        and~Andreas~Zuefle
\thanks{}
\thanks{T. Bernecker, T. Emrich, H.P. Kriegel, M. Renz, and A. Zuefle are with the Ludwig-Maximilians-Universit\"{a}t,
M\"{u}nchen, Germany. E-mail:
\{bernecker,emrich,kriegel,renz,zuefle\}@dbs.ifi.lmu.de. N.
Mamoulis is with the University of Hong Kong, Pokfulam Road, Hong
Kong. E-mail: nikos@cs.hku.hk.} }

\maketitle

\begin{abstract}
In this paper, we propose a novel, effective and efficient
probabilistic pruning criterion for probabilistic similarity
queries on uncertain data. Our approach supports a general
uncertainty model using continuous probabilistic density functions
to describe the (possibly correlated) uncertain attributes of
objects. In a nutshell, the problem to be solved is to compute the
PDF of the random variable denoted by the \emph{probabilistic
domination count}: Given an uncertain database object $B$, an
uncertain reference object $R$ and a set $\DB$ of uncertain
database objects in a multi-dimensional space, the probabilistic
domination count denotes the number of uncertain objects in $\DB$
that are closer to $R$ than $B$. This domination count can be used
to answer a wide range of probabilistic similarity queries.
Specifically, we propose a novel geometric pruning filter and
introduce an iterative filter-refinement strategy for
conservatively and progressively estimating the probabilistic
domination count in an efficient way while keeping correctness
according to the possible world semantics. In an experimental
evaluation, we show that our proposed technique allows to acquire
tight probability bounds for the probabilistic domination count
quickly, even for large uncertain databases.
\end{abstract}

\section{Introduction}\label{sec:introduction}

In the past two decades, there has been a great deal of interest
in developing efficient and effective methods for similarity
queries, e.g. $k$-nearest neighbor search, reverse $k$-nearest
neighbor search and ranking in spatial, temporal, multimedia and
sensor databases. Many applications dealing with such data have to
cope with uncertain or imprecise data.

In this work, we introduce a novel scalable pruning approach to
identify candidates for a class of probabilistic similarity
queries. Generally spoken, probabilistic similarity queries
compute for each database object $o\in\DB$ the probability that a
given query predicate is fulfilled. Our approach addresses
probabilistic similarity queries where the query predicate is
based on object (distance) relations, i.e. the event that an
object $B$ belongs to the result set depends on the relation of
its distance to the query object $R$ and the distance of another
object $A$ to the query object. Exemplarily, we apply our novel
pruning method to the most prominent queries of the above
mentioned class, including the probabilistic $k$-nearest neighbor
(P$k$NN) query, the probabilistic reverse $k$-nearest neighbor
(PR$k$NN) query and the probabilistic inverse ranking query.

\begin{figure}
    \centering
       \includegraphics[width=0.4\columnwidth]{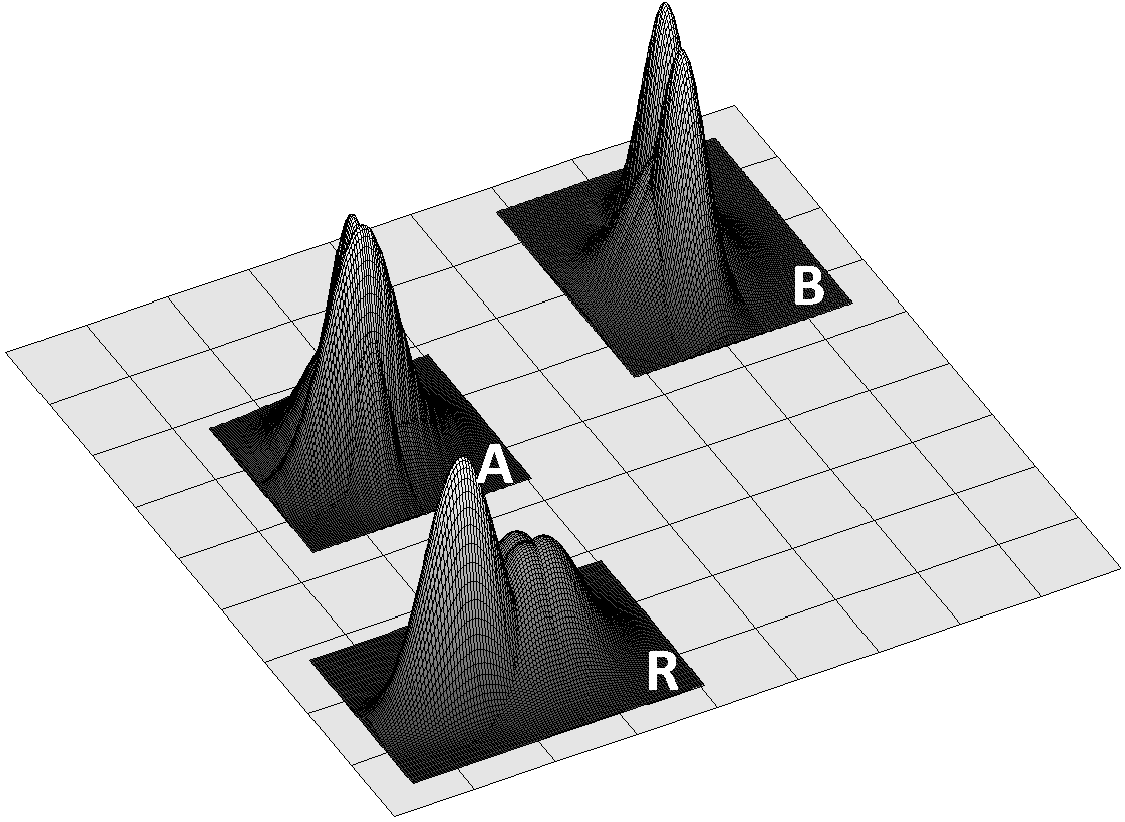}
    \caption{$A$ dominates $B$ w.r.t. $R$ with high probability.}
    \label{fig:basic_example}
\end{figure}

\subsection{Uncertainty Model}\label{sub:uncertaintyModel}

In this paper, we assume that the database $\DB$ consists of
multi-attribute objects $o_1,...,o_N$ that may have uncertain
attribute values. An uncertain attribute is defined as follows:
\begin{definition}[Probabilistic Attribute]
A probabilistic attribute $attr$ of object $o_i$ is a random
variable drawn from a probability distribution with density
function $f^{attr}_i$.
\end{definition}
An uncertain object $o_i$ has at least one uncertain attribute
value. The function $f_i$ denotes the multi-dimensional
probability density distribution (PDF) of $o_i$ that combines all
density functions for all probabilistic attributes $attr$ of
$o_i$.

Following the convention of uncertain databases
\cite{BesSolIly08,CheChe07,CheCheMokCho08,CheKalPra04,CorLiYi09,LianChen09,SinMayMitPraetal08},
we assume that $f_i$ is (minimally) bounded by an
\emph{uncertainty region} $R^i$ such that $\forall x \notin R^i:
f_i(x)=0$ and
$$
\int_{R^i}f_i(x)dx\leq 1.
$$
Specifically, the case $\int_{R^i}f_i(x)dx< 1$ implements
existential uncertainty, i.e. object $o_i$ may not exist in the
database at all with a probability greater than zero. In this
paper we focus on the case $\int_{R^i}f_i(x)dx = 1$, but the
proposed concepts can be easily adapted to existentially uncertain
objects. Although our approach is also applicable for unbounded
PDF, e.g., Gaussian PDF, here we assume $f_i$ exceeds zero only
within a bounded region. This is a realistic assumption because
the spectrum of possible values of attributes is usually bounded
and it is commonly used in related work, e.g.
\cite{CheChe07,CheCheMokCho08} and \cite{BesSolIly08}. Even if
$f_i$ is given as an unbounded PDF, a common strategy is to
truncate PDF tails with negligible probabilities and normalize the
resulting PDF. In specific, \cite{BesSolIly08} shows that for a
reasonable low truncation threshold, the impact on the accuracy of
probabilistic ranking queries is quite low while having a very
high impact on the query performance. In this way, each uncertain
object can be considered as a $d$-dimensional rectangle with an
associated multi-dimensional object PDF (c.f. Figure
\ref{fig:basic_example}). Here, we assume that uncertain
attributes may be mutually dependent. Therefore the object PDF can
have any arbitrary form, and in general, cannot simply be derived
from the marginal distribution of the uncertain attributes. Note
that in many applications, a discrete uncertainty model is
appropriate, meaning that the probability distribution of an
uncertain object is given by a finite number of alternatives
assigned with probabilities. This can be seen as a special case of
our model.

\subsection{Problem Formulation}\label{sub:problem}

We address the problem of detecting for a given uncertain object
$B$ the number of uncertain objects of an uncertain database $\DB$
that are closer to (\emph{i.e. dominate}) a reference object $R$
than $B$. We call this number the \emph{domination count} of $B$
w.r.t. $R$ as defined below:

\begin{definition}[Domination]
\label{def:Domination} Consider an uncertain database $\DB
=\{o_1,...,o_N\}$ and an uncertain reference object $R$. Let
$A,B\in \DB$. $Dom(A,B,R)$ is the random indicator variable that
is $1$, iff $A$ dominates $B$ w.r.t. $R$, formally:

$$
Dom(A,B,R) = \begin{cases}
  1,  & \text{if }dist(a,r)<dist(b,r) \\
  & \forall a \in A, b \in B, r \in R\\
  0, & \text{otherwise}
\end{cases}
$$
where $a,b$ and $r$ are samples drawn from the PDFs of $A,B$ and
$R$, respectively and $dist$ is a distance function on vector
objects.\footnote{We assume Euclidean distance for the remainder
of the paper, but the techniques can be applied to any $L_p$
norm.}
\end{definition}

\begin{definition}[Domination Count]
\label{def:DominationCount} Consider an uncertain database $\DB
=\{o_1,...,o_N\}$ and an uncertain reference object $R$. For each
uncertain object $B\in \DB$, let $DomCount(B,R)$ be the random
variable of the number of uncertain objects $A\in \DB$ ($A\neq B$)
that are closer to $R$ than $B$:
$$
DomCount(B,R)=\sum_{A\in \DB, A\neq B}Dom(A,B,R)
$$
\end{definition}
$DomCount(B,R)$ is the sum of $N-1$ non-necessarily identically
distributed and non-necessarily independent Bernoulli variables.
The problem solved in this paper is to efficiently compute the
probability density distribution of $DomCount(B,R) (B\in\DB)$
formally introduced by means of the probabilistic domination (cf.
Section \ref{sec:nnDomination}) and the probabilistic domination
count (cf. Section \ref{sec:DomCount}).

Determining domination is a central module for most types of
similarity queries in order to identify true hits and true drops
(pruning). In the context of probabilistic similarity queries,
knowledge about the PDF of $DomCount(B,R)$ can be used to find out
if $B$ satisfies the query predicate. For example, for a
probabilistic 5NN query with probability threshold $\tau = 10\%$
and query object $Q$, an object $B$ can be pruned (returned as a
true hit), if the probability $P(DomCount(B,Q)<5)$ is less (more)
than $10\%$.

\subsection{Overview}\label{sub:basicIdea}

Given an uncertain database $\DB =\{o_1,...,o_N\}$ and an
uncertain reference object $R$, our objective is to efficiently
derive the distribution of $DomCount(B,R)$ for any uncertain
object $B\in \DB$ and use it in the computation of probabilistic
similarity queries. First (Section \ref{sec:nnDomination}), we
build on the methodology of \cite{EmrKriKroRenZue10} to
efficiently find the complete set of objects in $\DB$ that
definitely dominate (are dominated by) $B$ w.r.t. $R$. At the same
time, we find the set of objects whose dominance relationship to
$B$ is uncertain. Using a decomposition technique, for each object
$A$ in this set, we can derive a lower and an upper bound for
$PDom(A,B,R)$, i.e., the probability that $A$ dominates $B$ w.r.t.
$R$. In Section \ref{sec:DomCount}, we show that due to
dependencies between object distances to $R$, these probabilities
cannot be combined in a straightforward manner to approximate the
distribution of $DomCount(B,R)$. We propose a solution that copes
with these dependencies and introduce techniques that help to to
compute the probabilistic domination count in an efficient way. In
particular, we prove that the bounds of $PDom(A,B,R)$ are mutually
independent if they are computed without a decomposition of $B$
and $R$. Then, we provide a class of uncertain generating
functions that use these bounds to build the distribution of
$DomCount(B,R)$. We then propose an algorithm which progressively
refines $DomCount(B,R)$ by iteratively decomposing the objects
that influence its computation (Section \ref{sec:implementation}).
Section \ref{sec:applications} shows how to apply this iterative
probabilistic domination count refinement process to evaluate
several types of probabilistic similarity queries. In Section
\ref{sec:experiments}, we experimentally demonstrate the
effectiveness and efficiency of our probabilistic pruning methods
for various parameter settings on artificial and real-world
datasets.

\section{Related Work}
\label{sec:related}

The management of uncertain data has gained increasing interest in
diverse application fields, e.g. sensor monitoring
\cite{CheSinPra05}, traffic analysis, location-based services
\cite{WolSisChaYes99} etc. Thus, modelling probabilistic databases
has become very important in the literature, e.g.
\cite{AntJanKocOlt07,SenDes07,SinMayMitPraetal08}. In general,
these models can be classified in two types: \emph{discrete} and
\emph{continuous} uncertainty models. \emph{Discrete models}
represent each uncertain object by a discrete set of alternative
values, each associated with a probability. This model is in
general adopted for probabilistic databases, where tuples are
associated with existential probabilities, e.g.
\cite{CorLiYi09,LiSahDes09,SolIly09,HuaPeiZhaLin08}.

In this work, we concentrate on the \emph{continuous} model in
which an uncertain object is represented by a probability density
function (PDF) within the vector space. In general, similarity
search methods based on this model involve expensive integrations
of the PDFs, hence special approximation and indexing techniques
for efficient query processing are typically employed
\cite{CheXiaPraShaVit04,TaoCheXiaNgaetal05}.

Uncertain similarity query processing has focused on various
aspects. A lot of existing work dealing with uncertain data
addresses probabilistic nearest neighbor (NN) queries for certain
query objects \cite{CheKalPra04,KriKunRen07} and for uncertain
queries \cite{IjiIsh09}. To reduce computational effort,
\cite{CheCheMokCho08} add threshold constraints in order to
retrieve only objects whose probability of being the nearest
neighbor exceeds a user-specified threshold to control the desired
confidence required in a query answer. Similar semantics of
queries in probabilistic databases are provided by Top-$k$ nearest
neighbor queries \cite{BesSolIly08}, where the $k$ most probable
results of being the nearest neighbor to a certain query point are
returned. Existing solutions on probabilistic $k$-nearest neighbor
($k$NN) queries restrict to expected distances of the uncertain
objects to the query object \cite{LjoSin07} or also use a
threshold constraint \cite{CheCheCheXie09}. However, the use of
expected distances does not adhere to the possible world semantics
and may thus produce very inaccurate results, that may have a very
small probability of being an actual result
(\cite{SolIly09,LiSahDes09}). Several approaches return the full
result to queries as a ranking of probabilistic objects according
to their distance to a certain query point
\cite{BerKriRen08,CorLiYi09,LiSahDes09,SolIly09}. However, all
these prior works have in common that the query is given as a
single (certain) point. To the best of our knowledge, $k$-nearest
neighbor queries as well as ranking queries on uncertain data,
where the query object is allowed to be uncertain, have not been
addressed so far. Probabilistic reverse nearest neighbor (RNN)
queries have been addressed in \cite{CheLinWanZhaPei10} to process
them on data based on discrete and continuous uncertainty models.
Similar to our solution, the uncertainty regions of the data are
modelled by MBRs. Based on these approximations, the authors of
\cite{CheLinWanZhaPei10} are able to apply a combination of
spatial, metric and probabilistic pruning criteria to efficiently
answer queries.

All of the above approaches that use MBRs as approximations for
uncertain objects utilize the minimum/maximum distance
approximations in order to remove possible candidates. However,
the pruning power can be improved using geometry-based pruning
techniques as shown in \cite{EmrKriKroRenZue10}. In this context,
\cite{LianChen09a} introduces a geometric pruning technique that
can be utilized to answer monochromatic and bichromatic
probabilistic RNN queries for arbitrary object distributions.

The framework that we introduce in this paper can be used to
answer probabilistic (threshold) $k$NN queries and probabilistic
reverse (threshold) $k$NN queries as well as probabilistic ranking
and inverse ranking queries for uncertain query objects.

\section{Similarity Domination on Uncertain Data}
\label{sec:nnDomination}

In this section, we tackle the following problem: Given three
uncertain objects $A$, $B$ and $R$ in a multidimensional space
$\mathbb{R}^d$, determine whether object $A$ is closer to $R$ than
$B$ w.r.t. a distance function defined on the objects in
$\mathbb{R}^d$. If this is the case, we say $A$ \emph{dominates}
$B$ w.r.t. $R$. In contrast to \cite{EmrKriKroRenZue10}, where
this problem is solved for certain data, in the context of
uncertain objects this domination relation is not a predicate that
is either true or false, but rather a (dichotomous) random
variable as defined in Definition \ref{def:Domination}. In the
example depicted in Figure \ref{fig:basic_example}, there are
three uncertain objects $A$, $B$ and $R$, each bounded by a
rectangle representing the possible locations of the object in
$\mathbb{R}^2$. The PDFs of $A$, $B$ and $R$ are depicted as well.
In this scenario, we cannot determine for sure whether object $A$
dominates $B$ w.r.t. $R$. However, it is possible to determine
that object $A$ dominates object $B$ w.r.t. $R$ with a high
probability. The problem at issue is to determine the
\emph{probabilistic domination probability} defined as:
\begin{definition}[Probabilistic Domination]
Given three uncertain objects $A$, $B$ and $R$, the probabilistic
domination $PDom(A,B,R)$ denotes the probability that $A$
dominates $B$ w.r.t. $R$.
\end{definition}

Naively, we can compute $PDom(A,B,R)$ by simply integrating the
probability of all possible worlds in which $A$ dominates $B$
w.r.t. $R$ exploiting inter-object independency:
$$
PDom(A,B,R)=\int_{a\in A}\int_{b\in B}\int_{r\in
R}\delta(a,b,r)\cdot P(A=a)\cdot P(B=b)\cdot P(R=r) da\, db\, dr,
$$
where $\delta(a,b,r)$ is the following indicator function:

$$
\delta (a,b,r) = \begin{cases}
  1,  & \text{if }dist(a,r)<dist(b,r)\\
  0, & \text{else}
\end{cases}
$$

The problem of this naive approach is the computational cost of
the triple-integral. The integrals of the PDFs of A, B and R may
in general not be representable as a closed-form expression and
the integral of $\delta(a,b,r)$ does not have a closed-from
expression. Therefore, an expensive numeric approximation is
required for this approach. In the rest of this section we propose
methods that efficiently derive bounds for $PDom(A,B,R)$, which
can be used to prune objects avoiding integral computations.

\subsection{Complete Domination}
\label{sub:spatialDomination}

First, we show how to detect whether $A$ completely dominates $B$
w.r.t. $R$ (i.e. if $PDom(A,B,R)=1$) regardless of the probability
distributions assigned to the rectangular uncertainty regions. The
state-of-the-art criterion to detect spatial domination on
rectangular uncertainty regions is with the use of minimum/maximum
distance approximations. This criterion states that $A$ dominates
$B$ w.r.t. $R$ if the minimum distance between $R$ and $B$ is
greater than the maximum distance between $R$ and $A$. Although
correct, this criterion is not tight (cf.
\cite{EmrKriKroRenZue10}), i.e. not each case where $A$ dominates
$B$ w.r.t. $R$ is detected by the min/max-domination criterion.
The problem is that the dependency between the two distances
between $A$ and $R$ and between $B$ and $R$ is ignored. Obviously,
the distance between $A$ and $R$ as well as the distance between
$B$ and $R$ depend on the location of $R$. However, since $R$ can
only have a unique location within its uncertainty region, both
distances are mutually dependent. Therefore, we adopt the spatial
domination concepts proposed in \cite{EmrKriKroRenZue10} for
rectangular uncertainty regions.

\begin{corollary}[Complete Domination]
\label{def:renzidist}
Let $A,B,R$ be uncertain objects with rectangular uncertainty
regions. Then the following statement holds:
$$
PDom(A,B,R) = 1 \Leftrightarrow \sum_{i=1}^d \max\limits_{r_i\in
\{R_i^{min}, R_i^{max}\}}(\MaxDist (A_i,r_i)^p - \MinDist
(B_i,r_i)^p) < 0,
$$
where $A_i,B_i$ and $R_i$ denote the projection interval of the
respective rectangular uncertainty region of $A$, $B$ and $R$ on
the $i^{th}$ dimension; $R_i^{min}$ $(R_i^{max})$ denotes the
lower (upper) bound of interval $R_i$, and $p$ corresponds to the
used $L_p$ norm. The functions $\MaxDist(A,r)$ and $\MinDist(A,r)$
denote the maximal (respectively minimal) distance between the
one-dimensional interval $A$ and the one-dimensional point $r$.
\end{corollary}
Corollary \ref{def:renzidist} follows directly from
\cite{EmrKriKroRenZue10}; the inequality is true if and only if
for all points $a\in A,b\in B,r\in R$, $a$ is closer to $r$ than
$b$. Translated into the possible worlds model, this is equivalent
to the statement that $A$ is closer to $R$ than $B$ for any
possible world, which in return means that $PDom(A,B,R) = 1$.

In addition, it holds that
\begin{corollary}
\label{cor:mutualexclusivenessofdomination}
$$
PDom(A,B,R) = 1 \Leftrightarrow PDom(B,A,R) = 0
$$
\end{corollary}

\begin{figure}
    \centering
    \subfigure[{\small Complete domination}
    \label{subfig:SpatialDomination}]{\includegraphics[width=0.26\columnwidth]{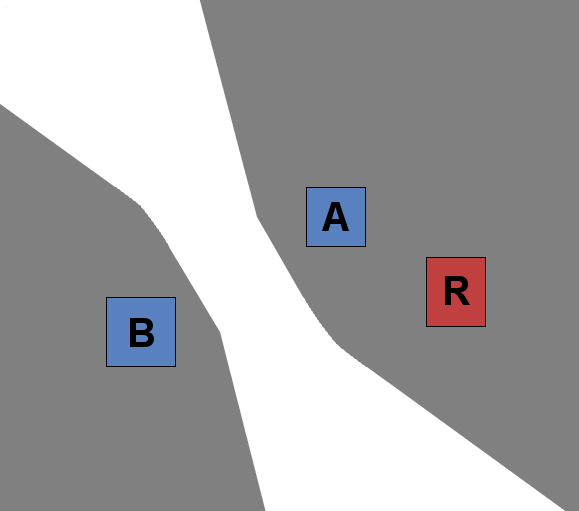}}\hspace{1cm}
    \subfigure[{\small Probabilistic domination}
    \label{subfig:ProbabilisticDomination}]{\includegraphics[width=0.3\columnwidth]{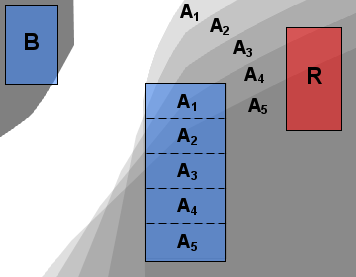}}
    \caption{Similarity Domination.}
    \label{fig:pruning_regions}
\end{figure}

In the example depicted in Figure \ref{subfig:SpatialDomination},
the grey region on the right shows all points that definitely are
closer to $A$ than to $B$ and the grey region on the left shows
all points that definitely are closer to $B$ than to $A$.
Consequently, $A$ dominates $B$ ($B$ dominates $A$) if $R$
completely falls into the right (left) grey shaded
half-space.\footnote{Note that the grey regions are not explicitly
computed; we only include them in Figure
\ref{subfig:SpatialDomination} for illustration purpose.}

\subsection{Probabilistic Domination}
\label{sub:probabilisticDomination}

Now, we consider the case where $A$ does not completely dominate
$B$ w.r.t. $R$. In consideration of the possible world semantics,
there may exist worlds in which $A$ dominates $B$ w.r.t. $R$, but
not all possible worlds satisfy this criterion. Let us consider
the example shown in Figure \ref{subfig:ProbabilisticDomination}
where the uncertainty region of $A$ is decomposed into five
partitions, each assigned to one of the five grey-shaded regions
illustrating which points are closer to the partition in $A$ than
to $B$. As we can see, $R$ only completely falls into three
grey-shaded regions. This means that $A$ does not completely
dominate $B$ w.r.t. $R$. However, we know that in some possible
worlds (at least in all possible words where $A$ is located in
$A_1$, $A_2$ or $A_3$) $A$ does dominate $B$ w.r.t. $R$. The
question at issue is how to determine the probability
$PDom(A,B,R)$ that $A$ dominates $B$ w.r.t. $R$ in an efficient
way. The key idea is to decompose the uncertainty region of an
object $X$ into subregions for which we know the probability that
$X$ is located in that subregion (as done for object $A$ in our
example). Therefore, if neither $Dom(A,B,R)$ nor $Dom(B,A,R)$
holds, then there may still exist subregions $A^\prime \subset A$,
$B^\prime\subset B$ and $R^\prime\subset R$ such that $A^\prime$
dominates $B^\prime$ w.r.t. $R^\prime$. Given disjunctive
decomposition schemes $\mathcal{\underline{A}}$,
 $\mathcal{\underline{B}}$ and $\mathcal{\underline{R}}$ we can
identify triples of subregions ($A^\prime \in
\mathcal{\underline{A}}$, $B^\prime \in \mathcal{\underline{B}}$,
$R^\prime\in \mathcal{\underline{R}}$) for which $Dom(A^\prime,
B^\prime, R^\prime)$ holds. Let $\delta(A^\prime, B^\prime,
R^\prime)$ be the following indicator function:

$$
\delta(A^\prime,B^\prime,R^\prime) = \begin{cases}
  1,  & \text{if }Dom(A^\prime,B^\prime,R^\prime)\\
  0, & \text{else}
\end{cases}
$$

\begin{lemma}
\label{lem:lb} \label{lem:decomposition} Let $A,B$ and $R$ be
uncertain objects with disjunctive object decompositions
$\mathcal{\underline{A}}, \mathcal{\underline{B}}$ and
$\mathcal{\underline{R}}$, respectively. To derive a lower bound
$PDom_{LB}(A,B,R)$ of the probability $PDom(A,B,R)$ that $A$
dominates $B$ w.r.t. $R$, we can accumulate the probabilities of
combinations of these subregions as follows:
$$
 PDom_{LB}(A,B,R)= \sum_{A^\prime \in
\mathcal{\underline{A}},B^\prime \in \mathcal{\underline{B}}
,R^\prime\in \mathcal{\underline{R}}}P(a\in A^\prime)\cdot P(b\in
B^\prime)\cdot P(r\in
R^\prime)\cdot\delta(A^\prime,B^\prime,R^\prime),
$$
where $P(X\in X^\prime)$ denotes the probability that object $X$
is located within the region $X^\prime$.
\end{lemma}
\begin{proof}
The probability of a combination $(A^\prime,B^\prime,R^\prime)$
can be computed by $P(a\in A^\prime)\cdot P(b\in B^\prime)\cdot
P(r\in R^\prime)$ due to the assumption of mutually independent
objects. These probabilities can be aggregated due to the
assumption of disjunctive subregions, which implies that any two
different combinations of subregions $(A^\prime \in
\mathcal{\underline{A}},B^\prime \in
\mathcal{\underline{B}},R^\prime\in \mathcal{\underline{R}})$ and
$(A^{\prime\prime} \in \mathcal{\underline{A}},B^{\prime\prime}
\in \mathcal{\underline{B}},R^{\prime\prime}\in
\mathcal{\underline{R}}$, $A^\prime \neq A^{\prime\prime} \vee
B^\prime \neq B^{\prime\prime} \vee R^\prime \neq
R^{\prime\prime}$ must represent disjunctive sets of possible
worlds. It is obvious that all possible worlds defined by
combinations $(A^\prime,B^\prime,R^\prime)$ where
$\delta(A^\prime,B^\prime,R^\prime)=1$, $A$ dominates $B$ w.r.t.
$R$. But not all possible worlds where $A$ dominates $B$ w.r.t.
$R$ are covered by these combinations and, thus, do not contribute
to $PDom_{LB}(A,B,R)$. Consequently, $PDom_{LB}(A,B,R)$ lower
bounds $PDom(A,B,R)$.
\end{proof}
Analogously, we can define an upper bound of $PDom(A,B,R)$:
\begin{lemma}
\label{lem:ub} An upper bound $PDom_{UB}(A,B,R)$ of $PDom(A,B,R)$
can be derived as follows:
$$
PDom_{UB}(A,B,R)=1-PDom_{LB}(B,A,R)
$$
\end{lemma}
Naturally, the more refined the decompositions are, the tighter
the bounds that can be computed and the higher the corresponding
cost of deriving them. 
In particular, starting from the entire MBRs of the objects, we
can progressively partition them to iteratively derive tighter
bounds for their dependency relationships until a desired degree
of certainty is achieved (based on some threshold). However, in
the next section, we show that the derivation of the domination
count $DomCount(B,R)$ of a given object $B$ (cf. Definition
\ref{def:DominationCount}), which is the main module of prominent
probabilistic queries cannot be straightforwardly derived with the
use of these bounds and we propose a methodology based on
generating functions for this purpose.

\section{Probabilistic Domination Count}
\label{sec:DomCount} In Section \ref{sec:nnDomination} we
described how to conservatively and progressively approximate the
probability that $A$ dominates $B$ w.r.t. $R$. Given these
approximations $PDom_{LB}(A,B,R)$ and $PDom_{UB}(A,B,R)$, the next
problem is to cumulate these probabilities to get an approximation
of the domination count $DomCount(B,R)$ of an object $B$ w.r.t.
$R$ (cf. Definition \ref{def:DominationCount}). To give an
intuition how challenging this problem is, we first present a
naive solution that can yield incorrect results due to ignoring
dependencies between domination relations in Section
\ref{subsec:dependencies}. To avoid the problem of dependent
domination relations, we first show in Section
\ref{subsec:MutuallyIndependentDomApprox} how to exploit object
independencies to derive domination bounds that are mutually
independent. Afterwards, in Section \ref{subsec:GFs}, we introduce
a new class of uncertain generating functions that can be used to
derive bounds for the domination count efficiently, as we show in
Section \ref{subsec:UGFs}. Finally, in Section
\ref{subsec:ProbabilisticDominationCountEstimation}, we show how
to improve our domination count approximation by considering
disjunct subsets of possible worlds for which a more accurate
approximation can be computed.

\begin{figure}[t]
    \centering
       \includegraphics[width=0.25\columnwidth]{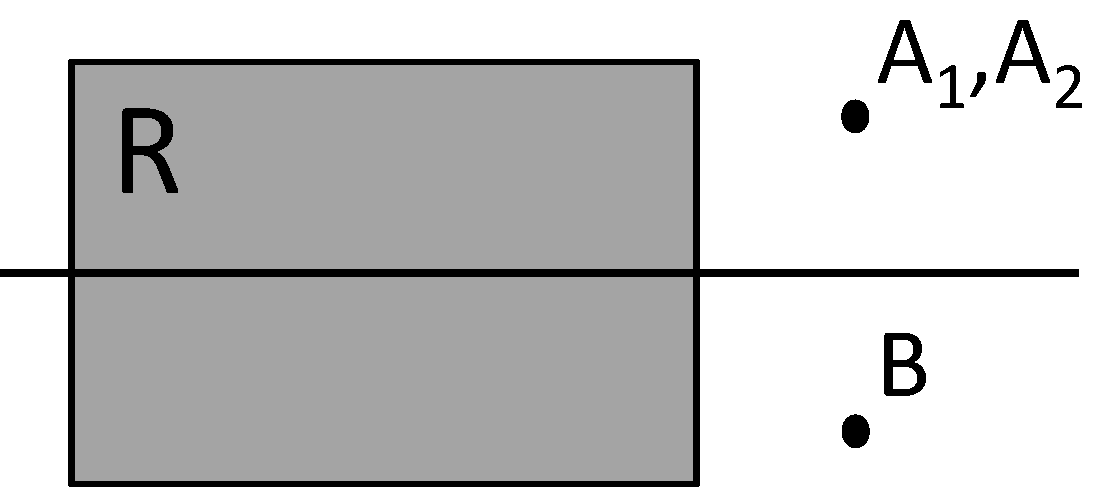}
    \caption{$A_1$ and $A_2$ dominate $B$ w.r.t. $R$ with a probability of 50\%, respectively.}
    \label{fig:dependency}
\end{figure}

\subsection{The Problem of Domination Dependencies}
\label{subsec:dependencies} To compute $DomCount(B,R)$, a
straightforward solution is to first approximate $PDom(A,B,R)$ for
all $A\in \DB$ using the technique proposed in Section
\ref{sec:nnDomination}. Then, given these probabilities we can
apply the technique of uncertain generating functions (cf. Section
\ref{GeneratingFunctions}) to approximate the probability that
exactly $0$, exactly $1$, ..., exactly $n-1$ uncertain objects
dominate $B$. However, this approach ignores possible dependencies
between domination relationships. Although we assume independence
between objects, the random variables $Dom(A_1,B,R)$ and
$Dom(A_2,B,R)$ are mutually dependent because the distance between
$A_1$ and $R$ depends on the distance between $A_2$ and $R$
because object $R$ can only appear once. Consider the following
example:
\begin{example}
\label{example:dependencies} Consider a database of three certain
objects $B$, $A_1$ and $A_2$ and the uncertain reference object
$R$, as shown in Figure \ref{fig:dependency}. For simplicity,
objects $A_1$ and $A_2$ have the same position in this example.
The task is to determine the domination count of $B$ w.r.t. $R$.
The domination half-space for $A_1$ and $A_2$ is depicted here as
well. Let us assume that $A_1$ ($A_2$) dominates $B$ with a
probability of $PDom(A_1,B,R)=PDom(A_2,B,R)=50\%$. Recall that
this probability can be computed by integration or approximated
with arbitrary precision using the technique of Section
\ref{sec:nnDomination}. However, in this example, the probability
that both $A_1$ and $A_2$ dominate $B$ is not simply $50\%\cdot
50\%=25\%$, as the generating function technique would return.

The reason for the wrong result in this example, is that the
generating function requires mutually independent random
variables. However, in this example, it holds that if and only if
$R$ falls into the domination half-space of $A_1$, it also falls
into the domination half-space of $A_2$. Thus we have the
dependency $dom(A_1,B,R) \leftrightarrow dom(A_2,B,R)$ and the
probability for $R$ to be dominated by both $A_1$ and $A_2$ is
$$
P(dom(A_1,B,R))\cdot P(dom(A_2,B,R)|dom(A_1,B,R))=0.5\cdot
1=0.5.$$
\end{example}

\subsection{Domination Approximations Based on Independent Objects}
\label{subsec:basicapprox} In general, domination relations may
have arbitrary correlations. Therefore, we present a way to
compute the domination count $DomCount(B,R)$ while accounting for
the dependencies between domination relations.
\subsubsection*{Complete Domination}
\label{subsec:MutuallyIndependentDomApprox} In an initial step,
\emph{complete domination} serves as a filter which allows us to
detect those objects $A\in\DB$ that definitely dominate a specific
object $B$ w.r.t. $R$ and those objects that definitely do not
dominate $B$ w.r.t. $R$ by means of evaluating $PDom(A,B,R)$. It
is important to note that complete domination relations are
mutually independent, since complete domination is evaluated on
the entire uncertainty regions of the objects. After applying
complete domination, we have detected objects that dominate $B$ in
all, or no possible worlds. Consequently, we get a first
approximation of the domination count $DomCount(B,R)$, obviously,
it must be higher than the number $N$ of objects that dominate $B$
and lower than $|\DB|-M$, where $M$ is the number of objects that
dominate $B$ in no possible world, i.e. $P(DomCount(B,R)=k)=0$ for
$k\leq N$ and $k\geq |\DB|-M$. Nevertheless, for $N<k<|\DB-M|$ we
still have a very bad approximation of the domination count
probability of $0\leq P(DomCount(B,R)=k)\leq 1$.

\subsubsection*{Probabilistic Domination}
 In order to refine this probability
distribution, we have to take the set of {\em influence} objects
$influenceObjects=\{A_1,...,A_C\}$, which neither completely prune
$B$ nor are completely dominated by $B$ w.r.t. $R$. For each
$A_i\in  influenceObjects$, $0<PDom(A_i,B,R)<1$. For these
objects, we can compute probabilities
$PDom(A_1,B,R),...,PDom(A_C,B,R)$ according to the methodology in
Section \ref{sec:nnDomination}. However, due to the mutual
dependencies between domination relations (cf. Section
\ref{subsec:dependencies}), we cannot simply use these
probabilities directly, as they may produce incorrect results.
However, we can use the observation that the objects $A_i$ are
mutually independent and each candidate object $A_i$ only appears
in a single domination relation $Dom(A_1,B,R),...,Dom(A_C,B,R)$.
Exploiting this observation, we can decompose the objects
$A_1,...,A_C$ only, to obtain mutually independent bounds for the
probabilities $PDom(A_1,B,R),...,PDom(A_C,B,R)$, as stated by the
following lemma:
\begin{lemma}\label{lem:independency} Let $A_1,...A_C$ be uncertain objects
with disjunctive object decompositions
$\mathcal{A}_1,...,\mathcal{A}_C$, respectively. Also, let $B$ and
$R$ be uncertain objects (without any decomposition). The lower
(upper) bound $PDom_{LB}(A_i,B,R)$ ($PDom_{UB}(A_i,B,R)$) as
defined in Lemma \ref{lem:lb} (Lemma \ref{lem:ub}) of the random
variable $Dom(A_i,B,R)$ is independent of the random variable
$Dom(A_j,B,R)$ ($1\leq i \neq j \leq C$).
\end{lemma}
\begin{proof}
Consider the random variable $Dom(A_i,B,R)$ conditioned on the
event $Dom(A_j,B,R)=1$. Using Equation \ref{lem:decomposition}, we
can derive the lower bound probability of
$Dom(A_i,B,R)=1|Dom(A_j,B,R)=1$ as follows:
$$
PDom_{LB}(A_i,B,R|Dom(A_j,B,R)=1)=$$ $$\sum_{A_i^\prime \in
\mathcal A_i,B^\prime \in \mathcal B,R^\prime\in \mathcal
R}[P(a_i\in A_i^\prime|Dom(A_j,B,R)=1)\cdot P(b\in
B^\prime|Dom(A_j,B,R)=1)\cdot
$$ \vspace{-0.6cm}
$$
P(r\in
R^\prime|Dom(A_j,B,R)=1)\cdot
\delta(A_i^\prime,B^\prime,R^\prime)]
$$\clearpage
Now we exploit that $B$ and $R$ are not decomposed, thus $B^\prime
= B$ and $R^\prime = R$, and thus $P(B\in
B^\prime|Dom(A_j,B,R)=1)=1=P(B\in B^\prime)$ and $P(R\in
R^\prime|Dom(A_j,B,R)=1)=1=P(R\in R^\prime)$. We obtain:
$$
PDom_{LB}(A_i,B,R|Dom(A_j,B,R)=1)=$$ $$\sum_{A_i^\prime \in
\mathcal A_i,B^\prime \in \mathcal B,R^\prime\in \mathcal
R}[P(a_i\in A_i^\prime|Dom(A_j,B,R)=1)\cdot P(b\in B^\prime)\cdot
P(r\in R^\prime)\cdot \delta(A_i^\prime,B^\prime,R^\prime)]
$$
Next we exploit that $P(a_i\in A_i^\prime|Dom(A_j,B,R)=1)=P(a_i\in
A_i^\prime)$ since $A_i$ is independent from $Dom(A_j,B,R)$ and
obtain:
$$
PDom_{LB}(A_i,B,R|Dom(A_j,B,R)=1)= $$ $$\sum_{A_i^\prime \in
\mathcal A_i,B^\prime \in \mathcal B,R^\prime\in \mathcal
R}[P(a_i\in A_i^\prime)\cdot P(b\in B^\prime)\cdot P(r\in
R^\prime)\cdot \delta(A_i^\prime,B^\prime,R^\prime)]
=PDom_{LB}(A_i,B,R)
$$

Analogously, it can be shown that
$$
PDom_{UB}(A_i,B,R|Dom(A_j,B,R)=1)=PDom_{UB}(A_i,B,R).
$$
\end{proof}

In summary, we can now derive, for each object $A_i$ a lower and
an upper bound of the probability that $A_i$ dominates $B$ w.r.t.
$R$. However, these bounds may still be rather loose, since we
only consider the full uncertainty region of $B$ and $R$ so far,
without any decomposition. In Section
\ref{subsec:ProbabilisticDominationCountEstimation}, we will show
how to obtain more accurate, still mutual independent probability
bounds based on decompositions of $B$ and $R$. Due to the mutual
independency of the lower and upper probability bounds, these
probabilities can now be used to get an approximation of the
domination count of $B$. In order to do this efficiently, we adapt
the generating functions technique which is proposed in
\cite{LiSahDes09}. The main challenge here is to extend the
generating function technique in order to cope with probability
bounds instead of concrete probability values. It can be shown
that a straightforward solution based on the existing generating
functions technique applied to the lower/upper probability bounds
in an appropriate way does solve the given problem efficiently,
but overestimates the domination count probability and thus, does
not yield good probability bounds. Rather, we have to redesign the
generating functions technique such that lower/upper probability
bounds can be handled correctly.

\subsection{Uncertain Generating Functions (UGFs)}
\label{subsec:GFs} In this subsection, we will give a brief survey
on the existing generating function technique (for more details
refer to \cite{LiSahDes09}) and then propose our new technique of
uncertain generating functions.

\subsubsection*{Generating Functions}
\label{sec:generatingFunctionsAppendix}

Consider a set of $N$ \emph{mutually independent}, but not
necessarily identically distributed Bernoulli $\{0,1\}$ random
variables $X_1,...,X_N$. Let $P(X_i)$ denote the probability that
$X_i=1$. The problem is to efficiently compute the sum
$$
\sum_{i=1}^N X_i = \sum_{i=1}^{N}Dom(A_i,B,R)
$$
of these random variables. A naive solution would be to count, for
each $0 \leq k \leq N$, all combinations with exactly $k$
occurrences of $X_i=1$ and accumulate the respective probabilities
of these combinations. This approach, however, shows a complexity
of $O(2^N)$. In \cite{BerKriRenVerZue09}, an approach was proposed
that achieves an $O(N)$ complexity using the \emph{Poisson
Binomial Recurrence}. Note that $O(N)$ time is asymptotically
optimal in general, since the computation involves at least $O(N)$
computations, namely $P(X_i), 1\leq i\leq N$. In the following, we
propose a different approach that, albeit having the same linear
asymptotical complexity, has other advantages, as we will see. We
apply the concept of generating functions as proposed in the
context of probabilistic ranking in \cite{LiSahDes09}. Consider
the function ${\cal{F}}(x)=\prod^n_{i=1}(a_i+b_ix)$. The
coefficient of $x^k$ in ${\cal{F}}(x)$ is given by:
$\sum_{|\beta|=k}\prod_{i:\beta_i=0}a_i\prod_{i:\beta_i=1}b_i$,
where $\beta=\langle \beta_1,...,\beta_N\rangle$ is a Boolean
vector, and $|\beta|$ denotes the number of 1's in $\beta$.

\label{sub:Efficient-Computation-of-Prob-Support} Now consider the
following generating function:

$$
{\cal{F}}^i=\prod_{X_i}(1-P(X_i)+P(X_i)\cdot x) = \sum_{j\geq
0}c_jx^j.
$$

The coefficient $c_j$ of $x^j$ in the expansion of ${\cal{F}}^i$
is the probability that for exactly $j$ random variables $X_i$ it
holds that $X_i=1$. Since ${\cal{F}}^i$ contains at most $i+1$
non-zero terms and by observing that
$$
{\cal{F}}^i={\cal{F}}^{i-1}\cdot (1-P(X_i)+P(X_i)\cdot x),
$$
we note that ${\cal{F}}^{i}$ can be computed in $O(i)$ time given
${\cal{F}}^{i-1}$. Since ${\cal{F}}^0=1x^0=1$, we conclude that
${\cal{F}}^{N}$ can be computed in $O(N^2)$ time. If only the
first $k$ coefficients are required (i.e. coefficients $c_j$ where
$j<k$), this cost can be reduced to $O(k\cdot N)$, by simply
dropping the summands $c_jx^j$ where $j\geq k$.

\begin{example}
As an example, consider three \emph{independent} random variables
$X_1$, $X_2$ and $X_3$. Let $P(X_1)=0.2$, $P(X_2)=0.1$ and
$P(X_3)=0.3$, and let $k=2$. Then:
$$
{\cal{F}}^1={\cal{F}}^0\cdot (0.8+0.2x)=0.2x^1 + 0.8x^0
$$
$$
{\cal{F}}^2={\cal{F}}^1\cdot
(0.9+0.1x)=0.02x^2+0.26x^1+0.72x^0\stackrel{*}{=}0.26x^1+0.72x^0
$$
$$
{\cal{F}}^3={\cal{F}}^2\cdot (0.7+0.3x)=0.078x^2+0.418x^1+0.504x^0
$$
$$
\stackrel{*}{=}0.418x^1+0.504x^0
$$ Thus, $P(DomCount(B)=0)=50.4\%$ and
$P(DomCount(B)=1)=41.8\%$. We obtain $P(DomCount(B)<2)=92.2\%$.
Thus, $B$ can be reported as a true hit if $\tau$ is not greater
than $92.2\%$. Equations marked by * exploit that we only need to
compute the $c_j$ where $j<k=2$.
\end{example}

\subsubsection*{Uncertain Generating Functions}
\label{GeneratingFunctions} Given a set of $N$ independent but not
necessarily identically distributed Bernoulli $\{0,1\}$ random
variables $X_i,1\leq i \leq N$. Let $P_{LB}(X_i)$ ($P_{UB}(X_i)$)
be a lower (upper) bound approximation of the probability
$P(X_i=1)$. Consider the random variable
$$
\sum_{i=1}^{N}X_i.
$$
We make the following observation: The lower and upper bound
probabilities $P_{LB}(X_i)$ and $P_{UB}(X_i)$ correspond to the
probabilities of the three following events:
\begin{itemize}
\item $X_i=1$ definitely holds with a probability of at least
$P_{LB}(Dom(A_i,B,R))$.

\item $X_i=0$ definitely holds with a probability of at least
$1-P_{UB}(X_i)$.

\item It is unknown whether $X_i=0$ or $X_i=1$ with the remaining
probability of
$P_{UB}(Dom(A_i,B,R))-P_{LB}(Dom(A_i,B,R))=PDom_{UB}(A_i,B,R)-PDom_{LB}(A_i,B,R)$.
\end{itemize}

Based on this observation, we consider the following uncertain
generating function (UGF):

$$
{\cal{F}}^N=\prod_{i \in 1,...,N}[(P_{LB}(X_i)\cdot x +
(1-P_{UB}(X_i))\cdot y+ (P_{UB}(X_i)-P_{LB}(X_i)))] =
\sum_{i,j\geq 0}c_{i,j}x^i y^j.
$$

The coefficient $c_{i,j}$ has the following meaning: With a
probability of $c_{i,j}$, $B$ is definitely dominated at least $i$
times, and possibly dominated another $0$ to $j$ times. Therefore,
the minimum probability that $\sum_{i=1}^{N}X_i=k$ is $c_{k,0}$,
since that is the probability that exactly $k$ random variables
$X_i$ are $1$. The maximum probability that $\sum_{i=1}^{N}X_i=k$
is $\sum_{i\leq k,i+j\geq k}c_{i,j}$, i.e. the total probability
of all possible combinations in which $\sum_{i=1}^{N}X_i=k$, may
hold. Therefore, we obtain an approximated PDF of
$\sum_{i=1}^{N}X_i$. In the approximated PDF of
$\sum_{i=1}^{N}X_i$, each probability $\sum_{i=1}^{N}X_i=k$ is
given by a conservative and a progressive approximation.

\begin{example}
\label{ex:UGF} Let $P_{LB}(X_1)=20\%$, $P_{UB}(X_1)=70\%$,
$P_{LB}(X_2)=60\%$ and $P_{UB}(X_2)=80\%$. The generating function
for the random variable $\sum_{i=1}^{2}X_i$ is the following:
$$
{\cal{F}}^2=(0.2x+0.5y+0.3)(0.6x+0.2y+0.2)=0.12x^2+0.34x+0.1+0.22xy+0.16y+0.06y^2
$$
That implies that, with a probability of at least $12\%$,
$\sum_{i=1}^{2}X_i=2$. In addition, with a probability of $22\%$
plus $6\%$, it may hold that $\sum_{i=1}^{2}X_i=2$, so that we
obtain a probability bound of $12\%-40\%$ for the random event
$\sum_{i=1}^{2}X_i=2$. Analogously, $\sum_{i=1}^{2}X_i=1$ with a
probability of $34\%-78\%$ and $\sum_{i=1}^{2}X_i=0$ with a
probability of $10\%-32\%$. The approximated PDF of
$\sum_{i=1}^{2}X_i$ is depicted in Figure
\ref{fig:approximatePDF}.
\end{example}

\begin{figure}[t]
    \centering
    \includegraphics[width = 0.45\columnwidth]{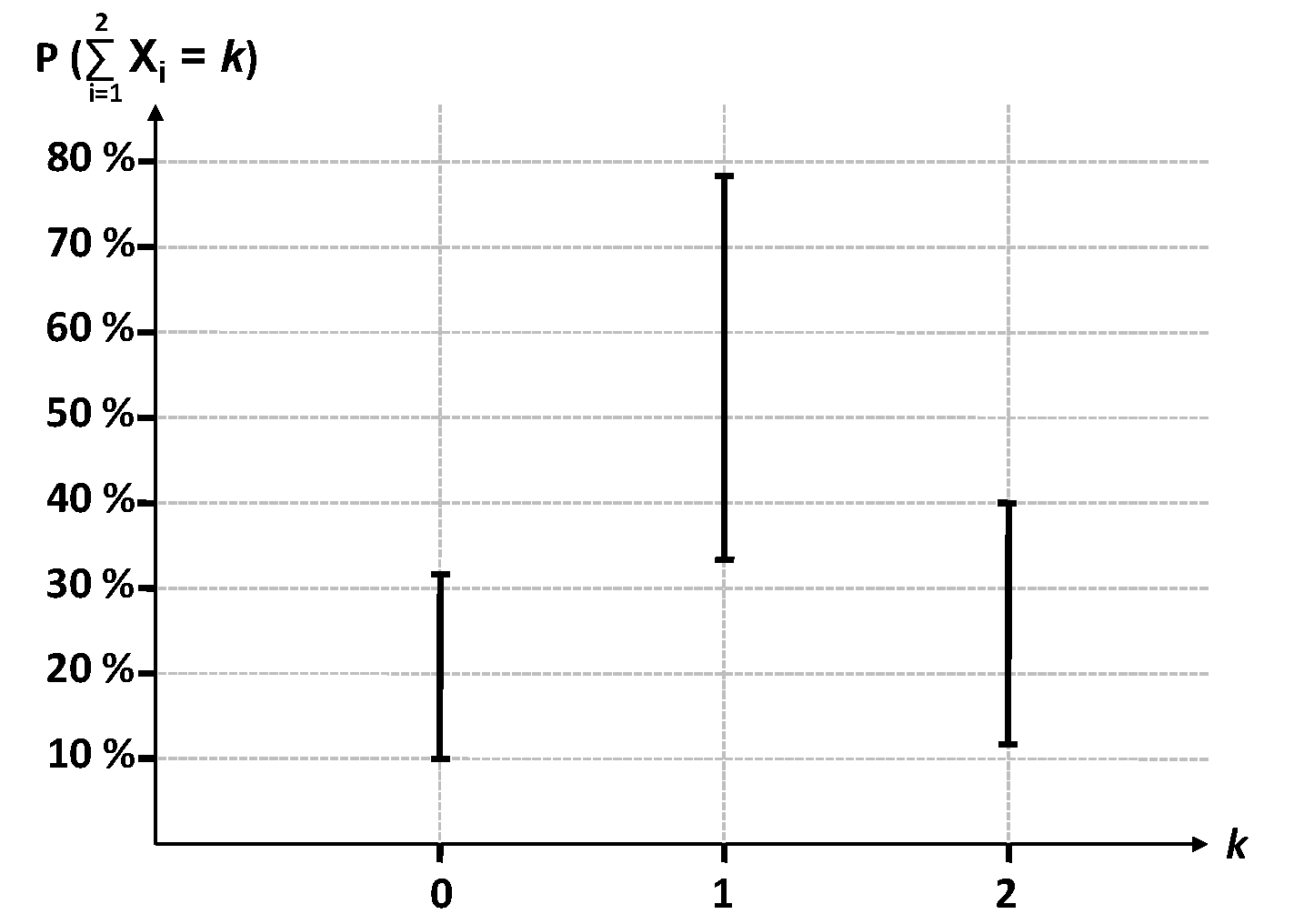}
    \caption{Approximated PDF of $\sum_{i=1}^{2}X_i$.}
  \label{fig:approximatePDF}
\end{figure}

Each expansion ${\cal{F}}^l$ can be obtained from the expansion of
${\cal{F}}^{l-1}$ as follows:

$${\cal{F}}^l={\cal{F}}^{l-1}\cdot$$
$$[P_{LB}(X_l)\cdot x + (1-P_{UB}(X_l)) + (P_{UB}(X_l)-P_{LB}(X_l))\cdot y].$$

We note that ${\cal{F}}^{l}$ contains at most $\sum_{i=1}^{l+1}i$
non-zero terms (one $c_{i,j}$ for each combination of $i$ and $j$
where $i+j\leq l$). Therefore, the total complexity to compute
${\cal{F}}^{l}$ is $O(l^3)$.

\subsection{Efficient Domination Count Approximation using UGFs} \label{subsec:UGFs}

We can directly use the uncertain generating functions proposed in
the previous section to derive bounds for the probability
distribution of the domination count $DomCount(B,R)$. Again, let
$\DB=A_1,...,A_N$ be an uncertain object database and $B$ and $R$
be uncertain objects in $\RR^d$. Let $Dom(A_i,B,R),1\leq i \leq N$
denote the random Bernoulli event that $A_i$ dominates $B$ w.r.t.
$R$.\footnote{That is, $X[Dom(A_i,B,R)]=1$ iff $A_i$ dominates $B$
w.r.t. $R$ and $X[Dom(A_i,B,R)]=0$ otherwise.} Also recall that
the domination count is defined as the random variable that is the
sum of the domination indicator variables of all uncertain objects
in the database (cf. Definition \ref{def:DominationCount}).

Considering the generating function
\begin{equation}
\label{eq:ugf} {\cal{F}}^N=\prod_{i \in
1,...,N}[(P_{LB}(Dom(A_i,B,R))\cdot x+
(P_{UB}(Dom(A_i,B,R))-P_{LB}(Dom(A_i,B,R)))\cdot
y)+$$\vspace{-0.6cm}
$$(1-P_{UB}(Dom(A_i,B,R)))] = \sum_{i,j\geq 0}c_{i,j}x^i y^j,
\end{equation}
we can efficiently compute lower and upper bounds of the
probability that $DomCount(B,R)=k$ for $0\leq k\leq|\DB|$, as
discussed in Section \ref{GeneratingFunctions} and because the
independence property of random variables required by the
generating functions is satisfied due to Lemma
\ref{lem:independency}.

\begin{lemma}
\label{lem:lbub} A lower bound $DomCount_{LB}^k(B,R)$ of the
probability that $DomCount(B,R)=k$ is given by
$$
DomCount_{LB}^k(B,R)=c_{k,0}
$$
and an upper bound $DomCount_{UB}^k(B,R)$ of the probability that
$DomCount(B,R)=k$ is given by
$$
DomCount_{UB}^k(B,R)=\sum_{i\leq k, i+j\geq k}c_{i,j}
$$
\end{lemma}

\begin{example}
Assume a database containing uncertain objects $A_1$, $A_2$, $B$
and $R$. The task is to determine a lower (upper) bound of the
domination count probability $DomCount_{LB}^k(B,R)$
($DomCount_{UB}^k(B,R)$) of $B$ w.r.t. $R$. Assume that, by
decomposing $A_1$ and $A_2$ and using the probabilistic domination
approximation technique proposed in Section
\ref{sub:probabilisticDomination}, we determine that $A_1$ has a
minimum probability $PDom_{LB}(A_1,B,R)$ of dominating $B$ of
$20\%$ and a maximum probability $PDom_{UB}(A_1,B,R)$ of $50\%$.
For $A_2$, $PDom_{LB}(A_2,B,R)$ is $60\%$ and $PDom_{UB}(A_2,B,R)$
is $80\%$. By applying the technique in the previous subsection,
we get the same generating function as in Example \ref{ex:UGF} and
thus, the same approximated PDF for the $DomCount(B,R)$ depicted
in Figure \ref{fig:approximatePDF}.
\end{example}

To compute the uncertain generating function and thus the
probabilistic domination count of an object in an uncertain
database of size $N$, the total complexity is $O(N^3)$. The reason
is that the maximal number of coefficients of the generating
function ${\cal{F}}^x$ is quadratic in $x$, since ${\cal{F}}^x$
contains coefficients $c_{i,j}$ where $i+j\leq x$, that is at most
$\frac{x^2}{2}$ coefficients. Since we have to compute
${\cal{F}}^x$ for each ($x<N$), the total time complexity is
$O(N^3)$. Note that only candidate objects $c\in Cand$ for which a
complete domination cannot be detected (cf. Section
\ref{sub:spatialDomination}) have to be considered in the
generating functions. Thus, the total runtime to compute
$DomCount_{LB}^k(B,R)$ as well as $DomCount_{UB}^k(B,R)$ is
$O(|Cand|^3)$. In addition, we will show in Section
\ref{sec:applications} how to reduce, specifically for $k$NN and
R$k$NN queries, the total time complexity to $O(k^2\cdot |Cand|)$.

\ifthenelse{\boolean{TR}} {

\subsection{Generating Functions vs Uncertain Generating Functions (Extended Version Only)}\label{sec:motivationGFAppendix}
It is clear that instead of applying the uncertain generating
function to approximate the domination count of $B$, two regular
generating functions can be used; one generating function that
uses the progressive (lower) bounds $P_{UB}(Dom(A_i,B,R))$ and one
that uses the conservative (upper) probability bounds
$P_{UB}(Dom(A_i,B,R))$. In the following we give an intuition and
a formal proof that using regular generating functions yields
looser bounds for the approximated domination.

Let $\DB=A_1,...,A_N$ be an uncertain object database and $B$ and
$R$ be uncertain objects in $\RR^d$. Let $Dom(A_i,B,R),1\leq i
\leq N$ denote the random Bernoulli event that $A_i$ dominates $B$
w.r.t. $R$. Let $P_{LB}(Dom(A_i,B,R))$ ($P_{UB}(Dom(A_i,B,R))$) be
a lower (upper) bound approximation of the probabilistic event
$X[Dom(A_i,B,R)]=1$.

A lower bound of the probability $DomCount(B,R)=k$ can be derived
using the following generating function:

\begin{equation}
\label{eq:gf_lb} {\cal{F}}^N=\prod_{i \in
1,...,N}[(P_{LB}(Dom(A_i,B,R))\cdot x + (1-P_{UB}(Dom(A_i,B,R)))]
= \sum_{i\geq 0}c_{i}x^i.
\end{equation}

Intuitively, this generating function uses the progressive
approximation $P_{LB}(Dom(A_i,B,R)$ of the probability that $A_i$
dominates $B$ w.r.t. $R$ and the progressive approximation
$1-P_{UB}(Dom(A_i,B,R))$ of the probability that $A_i$ does not
dominate $B$ w.r.t. $R$. This generating function is equal to the
uncertain generating function (cf. Equation \ref{eq:ugf}) if the
uncertain percentage (i.e. the coefficient of $y$) is omitted for
each candidate $A_i$, i.e. if the coefficient of each $y$ is set
to $0$.

\begin{lemma}
Let $c_k$ be the coefficients obtained by the generating function
in Equation \ref{eq:gf_lb} and let $P_{LB}(DomCount(B,R)=k)$ be
the lower bound derived by applying the generating function in
Equation \ref{eq:ugf} and exploiting Lemma \ref{lem:lbub}. It
holds that
$$
c_k=P_{LB}(DomCount(B,R)=k),
$$
i.e. the lower bound obtained by the (non-uncertain) generating
function in Equation \ref{eq:gf_lb} is as good as the lower bound
obtained using the technique in Section \ref{subsec:UGFs}.
\end{lemma}
\begin{proof}
The lower bound derived using the uncertain generating function
for the probability $DomCount(B,R)=k$ is equal to the coefficient
$c_{k,0}$. The coefficient $c_{k,0}$ corresponds to the variable
$x^ky^0=x^k$. Since Equation \ref{eq:gf_lb} and Equation
\ref{eq:ugf} are identical except for the variables containing at
least one $y$, but no $y$ is contained in the variable of the
coefficient $c_{k,0}$ in Equation \ref{eq:ugf}, it is identical to
the coefficient $c_k$ in Equation \ref{eq:gf_lb}.
\end{proof}

We can see that we can use a (non-uncertain) generating function
to derive the same lower bound. The advantage here is that the
(non-uncertain) generating function is easier to compute, due to a
much (linear in $k$) lower number of coefficients. However,
deriving an upper bound of the probability $DomCount(B,R)=k$ is
not as easy, because the upper bound does use the uncertainty of
objects.  An upper bound can be derived using the following
generating function:

\begin{equation}
\label{eq:gf_ub} {\cal{F}}^N=\prod_{i \in
1,...,N}[(P_{UB}(Dom(A_i,B,R))\cdot x+ (1-P_{LB}(Dom(A_i,B,R)))] =
\sum_{i\geq 0}c_{i}x^i.
\end{equation}

The idea is to use a conservative approximation
$P_{UB}(Dom(A_i,B,R)$ for the probability that $A_i$ dominates $B$
and a conservative approximation $1-P_{LB}(Dom(A_i,B,R))$ for the
probability that $A_i$ does not dominate $B$. The difference of
this generating function compared to the uncertain generating
function (cf. Equation \ref{eq:ugf}) is that the uncertain
percentage (the coefficient of $y$) is added to both probabilities
$Dom(A_i,B,R)$ and $\lnot (Dom(A_i,B,R)$\footnote{Note that adding
the coefficient of $y$ to only one of the summands will result in
incorrect bounds.}.

\begin{lemma}
Let $c_k$ be the coefficients obtained by the generating function
in Equation \ref{eq:gf_ub} and let $P_{LB}(DomCount(B,R)=k)$ be
the lower bound derived by applying the generating function in
Equation \ref{eq:ugf} and exploiting Lemma \ref{lem:lbub}. It
holds that
$$
c_k \geq P_{LB}(DomCount(B,R)=k),
$$
i.e. the upper bound obtained by the (non-uncertain) generating
function in Equation \ref{eq:gf_lb} is in general not as good as
the lower bound obtained using the technique in Section
\ref{subsec:UGFs}.
\end{lemma}

\begin{proof}
We give an example where the upper bound derived by Equation
\ref{eq:gf_lb} is worse than the upper bound derived by Equation
\ref{eq:ugf} and exploiting Lemma \ref{lem:lbub}. Assume a
database containing uncertain objects $A_1$, $A_2$, $B$ and $R$.
The task is to determine the probability that $DomCount(B,R)=1$.
Also assume that we have determined (e.g. as proposed in Section
\ref{sub:probabilisticDomination}) a lower bound
$P_{LB}(Dom(A_1,B,R))$ ($P_{LB}(Dom(A_2,B,R))$) and an upper bound
$P_{UB}(Dom(A_1,B,R))$ ($P_{UB}(Dom(A_2,B,R))$) of the probability
that $A_1$ ($A_2$) dominates $B$ w.r.t. $R$. To ease the notation,
let $A_i^+$ denote $P_{LB}(Dom(A_i,B,R))$, let $A_i^-$ denote
$1-P_{UB}(Dom(A_i,B,R))$ and let $A_i^?$ denote the uncertain
fraction $P_{UB}(Dom(A_i,B,R))-P_{LB}(Dom(A_i,B,R))$. Using this
notation, Equation \ref{eq:ugf} becomes:
$$
(A^+x+A^?y+A^-)\cdot (B^+x+B^?y+B^-)
$$
Expansion yields:
$$
A^+B^+x^2+A^+B^?xy+A^+B^-x+A^?B^+xy+A^?B^?y^2+
A^?B^-y+A^-B^+x+A^-B^?y+A^-B^-
$$
Exploiting Lemma \ref{lem:lbub} yields the following upper bound
probability for $DomCount(B,R)=1$:
$$
P_{UB}(DomCount(B,R)=1)=
A^+B^?+A^+B^-+A^?B^++A^?B^?+A^-B^?+A^-B^++A^-B^?
$$

On the other hand, Equation \ref{eq:gf_lb} becomes:
$$
((A^++A^?)x+A^-+A^?)\cdot ((B^++B^?)x+B^-+B^?)
$$
Expansion yields:
$$
(A^++A^?)\cdot(B^++B^?)x^2+(A^++A^?)\cdot (B^-+B^?)x +
(A^-+A^?)\cdot(B^++B^?)x+(A^-+A^?)\cdot(B^-+B^?)
$$
Extracting the coefficient $c_1$ of $x^1$ yields:
$$
c_1=(A^++A^?)\cdot (B^-+B^?)+(A^-+A^?)\cdot(B^++B^?)
$$
Expansion yields:
$$
c_1=A^+B^-+A^++B^?+A^?B^-+A^?B^?+ A^-B^++A^-B^?+A^?B^++A^?B^?
$$
Comparing $P_{UB}(DomCount(B,R)=1)$ and $c_1$, we obtain:
$$
c_1-P_{UB}(DomCount(B,R)=1)=A^?B^?
$$
Thus, the upper bound $c_1$ of the (non-uncertain) generating
function is greater (and thus worse) than the upper bound
$P_{UB}(DomCount(B,R)=1)$ of the uncertain generating function by
$A^?B^?$.
\end{proof}
Intuitively, the problem of the upper bound using the
(non-uncertain) generating function is that the uncertain fraction
(i.e. $A^?$) is added to both the probability that $A_i$ dominates
$B$ (i.e. $A^+$) and to the probability that $A_i$ does not
dominate $B$ (i.e. $A^-$). Therefore, this approach incorrectly
considers some possible worlds more often than once.

} {

\subsubsection*{Discussion}
In the extended version of this paper (\cite{ICDE2010extended}),
we show that instead of applying the uncertain generating function
to approximate the domination count of $B$, two regular generating
functions can be used; one generating function that uses the
progressive (lower) bounds $P_{UB}(Dom(A_i,B,R))$ and one that
uses the conservative (upper) probability bounds
$P_{UB}(Dom(A_i,B,R))$. However, we give an intuition and a formal
proof that using regular generating functions yields looser bounds
for the approximated domination. }

\subsection{Efficient Domination Count Approximation Based on Disjunctive Worlds}
\label{subsec:ProbabilisticDominationCountEstimation} Since the
uncertain objects $B$ and $R$ appear in each domination relation
$PDom(A_1,B,R)$,$...$,\linebreak$PDom(A_C,B,R)$ that is to
evaluate, we cannot split objects $B$ and $R$ independently (cf.
Section \ref{subsec:dependencies}). The reason for this dependency
is that knowledge about the predicate $Dom(A_i,B,R)$ may impose
constraints on the position of $B$ and $R$. Thus, for a partition
$B_1 \subset B$, the probability $PDom(A_j,B_1,R)$ may change
given $Dom(A_i,B,R)$ ($1\leq i,j \leq C, i\neq j$). However, note:
\begin{lemma}
\label{lem:part_independency} Given fixed partitions $B^\prime
\subseteq B$ and $R^\prime \subseteq R$, then the random variables
$Dom(A_i, B^\prime, R^\prime)$ are mutually independent for $1\leq
i,j \leq C, i\neq j$.
\end{lemma}
\begin{proof}
Similar to the proof of Lemma \ref{lem:independency}.
\end{proof}

This allows us to individually consider the subset of possible
worlds where $b\in B^\prime$ and $r\in R^\prime$ and use Lemma
\ref{lem:part_independency} to efficiently compute the
approximated domination count probabilities
$DomCount_{LB}^k(B^\prime,R^\prime)$ and
$DomCount_{UB}^k(B^\prime,R^\prime)$ under the condition that $B$
falls into a partition $B^\prime \subseteq B$ and $R$ falls into a
partition $R^\prime \subseteq R$. This can be performed for each
pair $(B^\prime,R^\prime)\in \underline{\mathcal{B}} \times
\underline{\mathcal{R}}$, where $\underline{\mathcal{B}}$ and
$\underline{\mathcal{R}}$ denote the decompositions of $B$ and
$R$, respectively. Now, we can treat pairs of partitions
$(B^\prime,R^\prime)\in \underline{\mathcal{B}} \times
\underline{\mathcal{R}}$ independently, since all pairs of
partition represent disjunctive sets of possible worlds due to the
assumption of a disjunctive partitioning. Exploiting this
independency, the PDF of the domination count $DomCount(B,R)$ of
the total objects $B$ and $R$ can then be obtained by creating an
uncertain generating function for each pair $(B^\prime,R^\prime)$
to derive a lower and an upper bound of
$P(DomCount(B^\prime,R^\prime)=k)$ and then computing the weighted
sum of these bounds as follows:
$$
DomCount_{LB}^k(B,R)= \sum_{B^\prime \in
\underline{\mathcal{B}},R^\prime \in
\underline{\mathcal{R}}}DomCount_{LB}^k(B^\prime,R^\prime)\cdot
P(B^\prime)\cdot P(R^\prime).
$$
The complete algorithm of our domination count approximation
approach can be found in the next Section.

\section{Implementation}
\label{sec:implementation} Algorithm \ref{alg:pruning} is a
complete method for iteratively computing and refining the
probabilistic domination count for a given object $B$ and a
reference object $R$. The algorithm starts by detecting complete
domination (cf. Section \ref{sub:spatialDomination}). For each
object that completely dominates $B$, a counter
$CompleteDominationCount$ is increased and each object that is
completely dominated by $B$ is removed from further consideration,
since it has no influence on the domination count of $B$. The
remaining objects, which may have a probability greater than zero
and less than one to dominate $B$, are stored in a set
$influenceObjects$. The set $influenceObjects$ is now used to
compute the probabilistic domination count ($DomCount_{LB}$,
$DomCount_{UB}$)\footnote{$DomCount_{LB}$ and $DomCount_{UB}$ are
lists containing, at each position $i$, a lower and an upper bound
for $P(DomCount(B,R)=i)$, respectively. This notation is
equivalent to a single uncertain domination count PDF.}: The main
loop of the probabilistic domination count approximation starts in
line 14. In each iteration, $B$, $R$, and all  influence objects
are partitioned. For each combination of partitions $B^\prime$ and
$R^\prime$, and each database object $A_i \in influenceObjects$,
the probability $PDom(A_i,B^\prime,R^\prime)$ is approximated (cf.
Section \ref{subsec:MutuallyIndependentDomApprox}). These
domination probability bounds are used to build an uncertain
generating function (cf. Section \ref{subsec:UGFs}) for the
domination count of $B^\prime$ w.r.t. $R^\prime$. Finally, these
domination counts are aggregated for each pair of partitions
$B^\prime, R^\prime$ into the domination count $DomCount(B,R)$
(cf. Section \ref{subsec:ProbabilisticDominationCountEstimation}).
The main loop continues until a domain- and user-specific stop
criterion is satisfied. For example, for a threshold $k$NN query,
a stop criterion is to decide whether the lower (upper) bound that
$B$ has a domination count of less than (at least) $k$, exceeds
(falls below) the given threshold.

The progressive decomposition of objects (line 15) can be
facilitated by precomputed split points at the object PDFs. More
specifically, we can iteratively split each object $X$ by means of
a median-split-based bisection method and use a kd-tree
\cite{BenLou75} to hierarchically organize the resulting
partitions. The kd-tree is a binary tree. The root of a kd-tree
represents the complete region of an uncertain object. Every node
implicitly generates a splitting hyperplane that divides the space
into two subspaces. This hyperplane is perpendicular to a chosen
split axis and located at the median of the node's distribution in
this axis. The advantage is that, for each node in the kd-tree,
the probability of the respective subregion $X^{\prime}$ is simply
given by $0.5^{X^{\prime}.level-1}$, where $X^{\prime}.level$ is
the level of $X^\prime$. In addition, the bounds of a subregion
$X^{\prime}$ can be determined by backtracking to the root. In
general, for continuously partitioned uncertain objects, the
corresponding kd-tree may have an infinite height, however for
practical reasons, the height $h$ of the kd-tree is limited. The
choice of $h$ is a trade-off between approximation quality and
efficiency: for a very large $h$, considering each leaf node is
similar to applying integration on the PDFs, which yields an exact
result; however, the number of leaf nodes, and thus the worst case
complexity increases exponentially in $h$. Note that our
experiments (c.f. Section \ref{sec:experiments}) show that a low
$h$ value is sufficient to yield reasonably tight approximation
bounds. Yet it has to be noted, that in the general case of
continuous uncertainty, our proposed approach may only return an
approximation of the exact probabilistic domination count.
However, such an approximation may be sufficient to decide a given
predicate as we will see in Section \ref{sec:applications} and
even in the case where the approximation does not suffice to
decide the query predicate, the approximation will give the user a
confidence value, based on which a user may be able decide whether
to include an object in the result.

\begin{algorithm}[tbh]
\footnotesize
  \caption{Probabilistic Inverse Ranking}
  \label{alg:pruning}
    \begin{algorithmic}[1]
        \REQUIRE: $Q$, $B$, $\DB$
        \STATE $influenceObjects = \emptyset$
        \STATE $CompleteDominationCount = 0$
        \STATE \textit{//Complete Domination}
        \FORALL{$A_i\in\DB$}
            \IF{$DDC_{Optimal}(A_i,B,R)$}
            \STATE $CompleteDominationCount$++
            \ELSIF{$\neg DDC_{Optimal}(B,A_i,R)$}
                \STATE $influenceObjects =influenceObjects \cap A_i$
            \ENDIF
        \ENDFOR
        \STATE \textit{//probabilistic domination count}
        \STATE $DomCount_{LB}$= [0,...,0] //length $|\DB|$
        \STATE $DomCount_{UB}$= [1,...,1] //length $|\DB|$
        \WHILE{$\neg$ stopcriterion}
            \STATE split($R$), split($B$), split($A_i \in \DB$)

            \FORALL{$B^\prime \in B$, $R^\prime\in R$}
                \STATE $cand_{LB}$= [0,...,0] //length $|uncertainObjects|$
                \STATE $cand_{UB}$= [1,...,1] //length $|uncertainObjects|$
                \FORALL{($0<i<|influenceObjects|$)}
                    \STATE $A_i=influenceObjects[i]$
                    \FORALL{$A_i^\prime \in A_i$}
                        \IF{$DDC_{Optimal}(A_i^\prime,B^\prime,R^\prime)$}
                            \STATE $cand_{LB}[i]$+=$(P(A_i^\prime))$
                        \ELSIF{$DDC_{Optimal}(B^\prime,A_i^\prime,R^\prime)$}
                            \STATE $cand_{UB}[i]$-=$(P(A_i^\prime))$
                        \ENDIF
                    \ENDFOR
                \ENDFOR
                \STATE compute $DomCount_{LB}(B^\prime,R^\prime)$ and $DomCount_{UB}(B^\prime,R^\prime)$
                using UGFs.
                \FORALL{ ($0<i<\DB$)}
                    \STATE $DomCount_{LB}[i]$+=$DomCount(B^\prime,R^\prime)_{LB}\cdot P(B^\prime)\cdot P(R^\prime)$
                    \STATE $DomCount_{UB}[i]$+=$DomCount(B^\prime,R^\prime)_{UB}\cdot P(B^\prime)\cdot P(R^\prime)$
                \ENDFOR
            \ENDFOR
            \STATE ShiftRight($DomCount_{LB}$,$CompleteDominationCount$)
            \STATE ShiftRight($DomCount_{UB}$,$CompleteDominationCount$)
        \ENDWHILE
        \STATE return ($DomCount_{LB}$, $DomCount_{UB}$)
    \end{algorithmic}
\end{algorithm}

\section{Applications}
\label{sec:applications} In this section, we outline how the
probabilistic domination count can be used to efficiently evaluate
a variety of probabilistic similarity query types, namely the
probabilistic inverse similarity ranking query \cite{LianChen09},
the probabilistic threshold $k$-NN query \cite{CheCheCheXie09},
the probabilistic threshold reverse $k$-NN query and the
probabilistic similarity ranking query
\cite{BerKriRen08,CorLiYi09,LiSahDes09,SolIly09}. We start with
the probabilistic inverse ranking query, because it can be derived
trivially from the probabilistic domination count introduced in
Section \ref{sec:DomCount}. In the following, let
$\DB=\{A_1,...,A_N\}$ be an uncertain database containing
uncertain objects $A_1,...,A_N$.

\begin{corollary}
Let $B$ and $R$ be uncertain objects. The task is to determine the
probabilistic ranking distribution $Rank(B,R)$ of $B$ w.r.t. to
similarity to $R$, i.e. the distribution of the position
$Rank(B,R)$ of object $B$ in a complete similarity ranking of
$A_1,...,$ $A_N,B$ w.r.t. the distance to an uncertain reference
object $R$. Using our techniques, we can compute $Rank(B,R)$ as
follows:
$$
P(Rank(B,R)=i)=P(DomCount(B,R)=i-1)
$$
\end{corollary}

The above corollary is evident, since the proposition ``$B$ has
rank $i$'' is equivalent to the proposition ``$B$ is dominated by
$i-1$ objects''.

The most prominent probabilistic similarity search query is the
probabilistic threshold $k$NN query.
\begin{corollary}
Let $Q=R$ be an uncertain query object and let $k$ be a scalar.
The problem is to find all uncertain objects $kNN_\tau(Q)$ that
are the $k$-nearest neighbors of $Q$ with a probability of at
least $\tau$. Using our techniques, we can compute the probability
$P^{kNN}(B,Q)$ that an object $B$ is a $k$NN of $Q$ as follows:
\label{knn}
$$
P^{kNN}(B,Q)=\sum_{i=0}^{k-1}P(DomCount(B,Q)=i)
$$
\end{corollary}

The above corollary is evident, since the proposition ``$B$ is a
$k$NN of $Q$'' is equivalent to the proposition ``$B$ is dominated
by less than $k$ objects''. To decide whether $B$ is a $k$NN of
$Q$, i.e. if $B\in kNN_\tau(Q)$, we just need to check if
$P^{kNN}(B,Q)>\tau$.

Next we show how to answer probabilistic threshold R$k$NN queries.

\begin{corollary}
\label{rknn} Let $Q=R$ be an uncertain query object and let $k$ be
a scalar. The problem is to find all uncertain objects $A_i$ that
have $Q$ as one of their $k$NNs with a probability of at least
$\tau$, that is, all objects $A_i$ for which it holds that $Q \in
kNN_\tau(A_i)$. Using our techniques, we can compute the
probability $P^{RkNN}(B,Q)$ that an object $B$ is a R$k$NN of $Q$
as follows:
$$
P^{RkNN}(B,Q)=\sum_{i=0}^{k-1}P(DomCount(Q,B)=i)
$$
\end{corollary}

The intuition here is that an object $B$ is a R$k$NN of $Q$ if and
only if $Q$ is dominated less than $k$ times w.r.t. $B$.

For $k$NN and R$k$NN queries, the total complexity to compute the
uncertain generating function can be improved from $O(|Cand|^3)$
to $O(|Cand|\cdot k^2)$ since it can be observed from Corollaries
\ref{knn} and \ref{rknn} that for $k$NN and R$k$NN queries, we
only require the section of the PDF of $DomCount(B,R)$ where
$DomCount(B,R)<k$, i.e. we only need to know the probabilities
$P(DomCount(B,R)=x), x<k$. This can be exploited to improve the
runtime of the computation of the PDF of $DomCount(B,R)$ as
follows: Consider the iterative computation of the generating
functions $\mathcal{F}^1,...,\mathcal{F}^{|cand|}$. For each
$\mathcal{F}^l, 1\leq l\leq |cand|$, we only need to consider the
coefficients $c_{i,j}$ in the generating function $\mathcal{F}^i$
where $i<k$, since only these coefficients have an influence on
$P(DomCount(B,R)=x), x<k$ (cf. Section \ref{lem:lbub}). In
addition, we can merge all coefficients $c_{i,j}$,
$c_{i^\prime,j^\prime}$ where $i=i^\prime$, $i+j>k$ and
$i^\prime+j^\prime>k$, since all these coefficients only differ in
their influence on the upper bounds of $P(DomCount(B,R)=x), x\geq
k$, and are treated equally for $P(DomCount(B,R)=x), x<k$. Thus,
each $\mathcal{F}^l$ contains at most $\sum_{i=1}^{k+1}i$
coefficients (one $c_{i,j}$ for each combination of $i$ and $j$
where $i+j\leq k$). Thus reducing the total complexity to
$O(k^2\cdot |cand|)$.

Finally, we show how to compute the expected rank (cf.
\cite{CorLiYi09}) of an uncertain object.
\begin{corollary}
Let $Q=R$ be an uncertain query object. The problem is to rank the
uncertain objects $A_i$ according to their expected rank
$E(Rank(A_i))$ w.r.t. the distance to $Q$. The expected rank of an
uncertain object $A_i$ can be computed as follows:
$$
E(Rank(A_i))=\sum_{i=0}^{N-1}P(DomCount(Q,B)=i)\cdot (i+1)
$$
\end{corollary}

Other probabilistic similarity queries (e.g. $k$NN and R$k$NN
queries with a different uncertainty predicate instead of a
threshold $\tau$) can be approximated efficiently using our
techniques as well. Details are omitted due to space constraints.

\section{Experimental Evaluation}
\label{sec:experiments}

In this section, we review the characteristics of the proposed
algorithm on synthetic and real-world data. The algorithm will be
referred to as \textbf{IDCA} (Iterative Domination Count
Approximation). We performed experiments under various parameter
settings. Unless otherwise stated, for 100 queries, we chose $B$
to be the object with the $10^{th}$ smallest MinDist to the
reference object $R$. We used a synthetic dataset with 10,000
objects modeled as 2D rectangles. The degree of uncertainty of the
objects in each dimension is modeled by their relative extent. The
extents were generated uniformly and at random with 0.004 as
maximum value. For the evaluation on real-world data, we utilized
the International Ice Patrol (IIP) Iceberg Sightings
Dataset\footnote{The IIP dataset is available at the National Snow
and Ice Data Center (NSIDC) web site
(\emph{http://nsidc.org/data/g00807.html}).}. This dataset
contains information about iceberg activity in the North Atlantic
in 2009. The latitude and longitude values of sighted icebergs
serve as certain 2D mean values for the 6,216 probabilistic
objects that we generated. Based on the date and the time of the
latest sighting, we added Gaussian noise to each object, such that
the passed time period since the latest date of sighting
corresponds to the degree of uncertainty (i.e. the extent). The
extents were normalized w.r.t. the extent of the data space, and
the maximum extent of an object in either dimension is 0.0004.

\begin{figure}[t]
    \centering
    \includegraphics[width = 0.4\columnwidth]{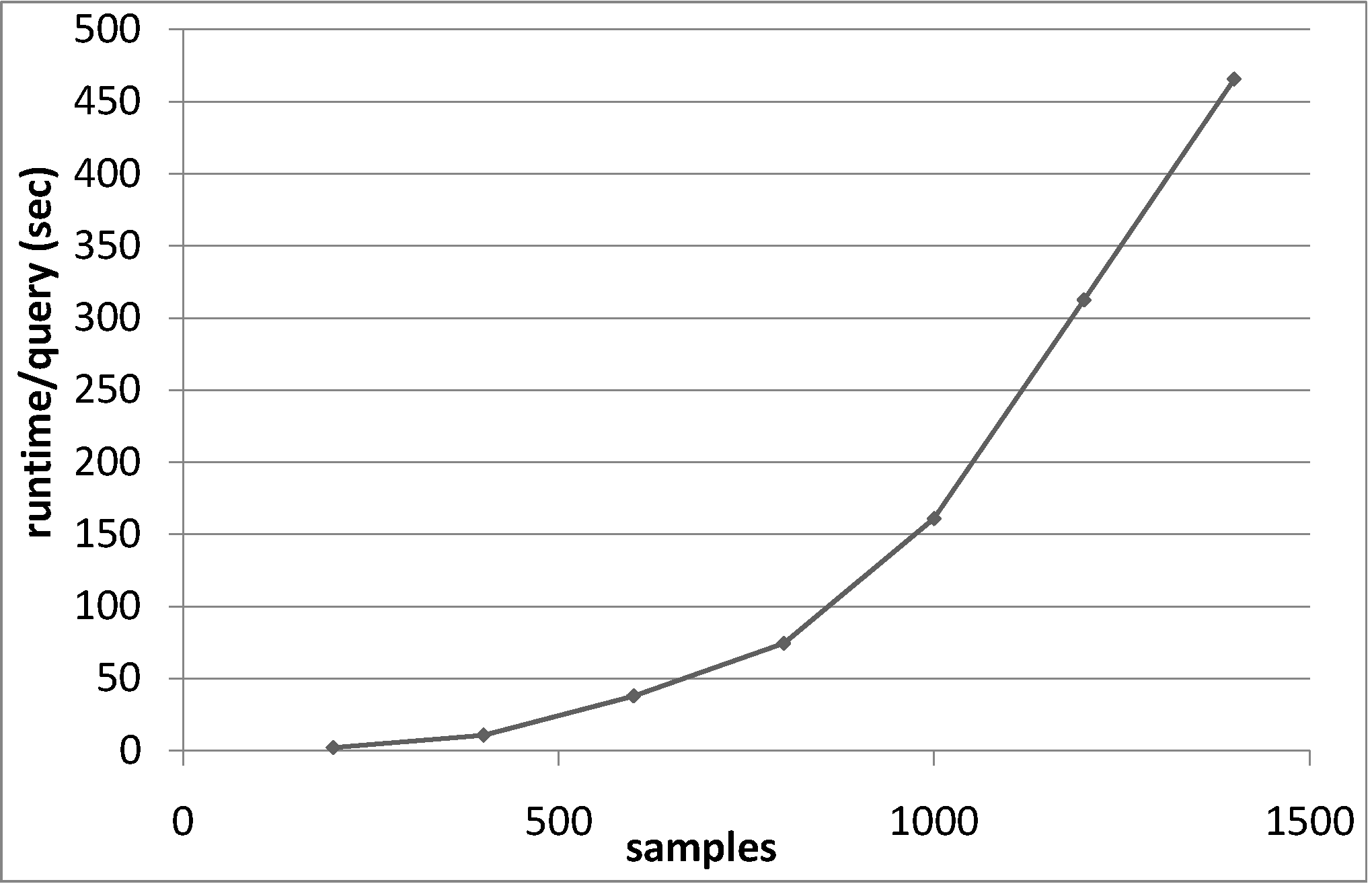}
    \caption{Runtime of \textbf{MC} for increasing sample size.}
  \label{fig:naive}
\end{figure}

\subsection{Runtime of the Monte-Carlo-based Approach}
\label{subsec:montecarlo} To the best of our knowledge, there
exists no approach which is able to process uncertain similarity
queries on probabilistic databases with continuous PDFs. A naive
approach needs to consider all possible worlds and thus needs to
integrate over all object PDFs, implying a runtime exponentially
in the number of objects. Since this is not applicable even for
small databases, we adapted an existing approach to cope with the
conditions. The approach most related to our work is
\cite{LianChen09}, which solves the problem of computing the
domination count for a certain query and discrete distributions
within the database objects. Thus the proposed comparison partner
works as follows: Draw a sufficiently large number $S$ of samples
from each object by Monte-Carlo-Sampling. Then, for each sample
$q_i \in Q$ of the query, apply the algorithm proposed in
\cite{LianChen09} to compute an exact probabilistic domination
count PDF of an object $B$. As proposed in \cite{LianChen09}, this
is done using the generating function technique and using an
\emph{and/xor tree} to combine individual samples into discrete
distributed uncertain objects. Finally, accumulate the resulting
certain domination count PDFs of each $q_i \in Q$ into a single
domination count PDF by taking the average. The execution time for
this approach, which we will refer to as \textbf{MC} in the
following, is shown in Figure \ref{fig:naive}. It can be observed
that for a reasonable sample size (which is required to achieve a
result that is close to the correct result with high probability)
the runtime becomes very large.

Note that our comparison partner only works for discrete uncertain
data (cf. Section~\ref{subsec:montecarlo}). To make a fair
comparison our approach relies on the same uncertainty model
(default: 1000 samples/object). Nevertheless, all the experiments
yield analogous results for continuous distributions.

\begin{figure}[t]
    \centering
    \subfigure[Candidates after spatial pruning.]{
        \label{fig:cands}
        \includegraphics[width = 0.4\columnwidth]{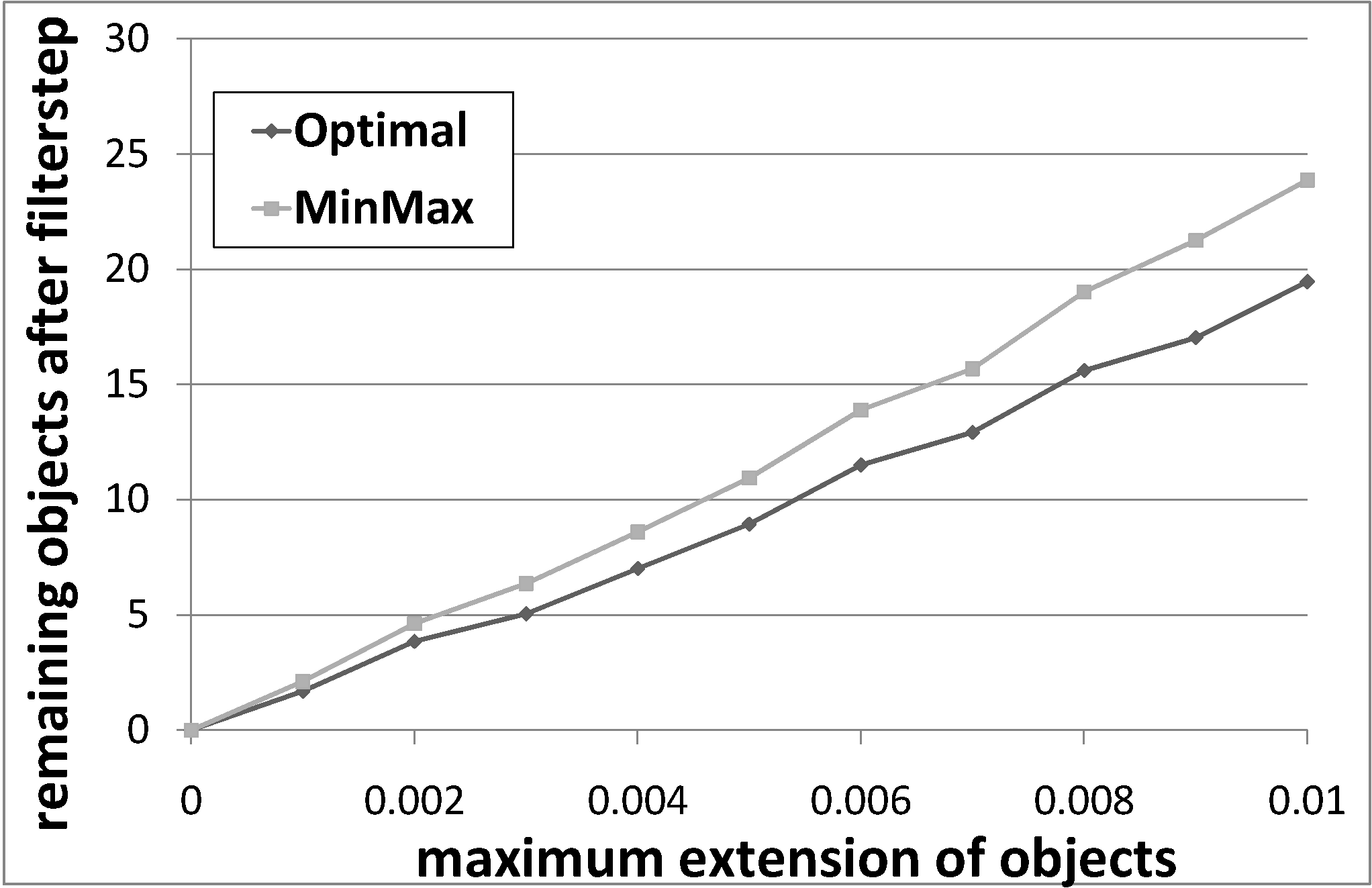}
    }
    \subfigure[Accumulated uncertainty of result.]{
        \label{fig:optimal}
        \includegraphics[width = 0.4\columnwidth]{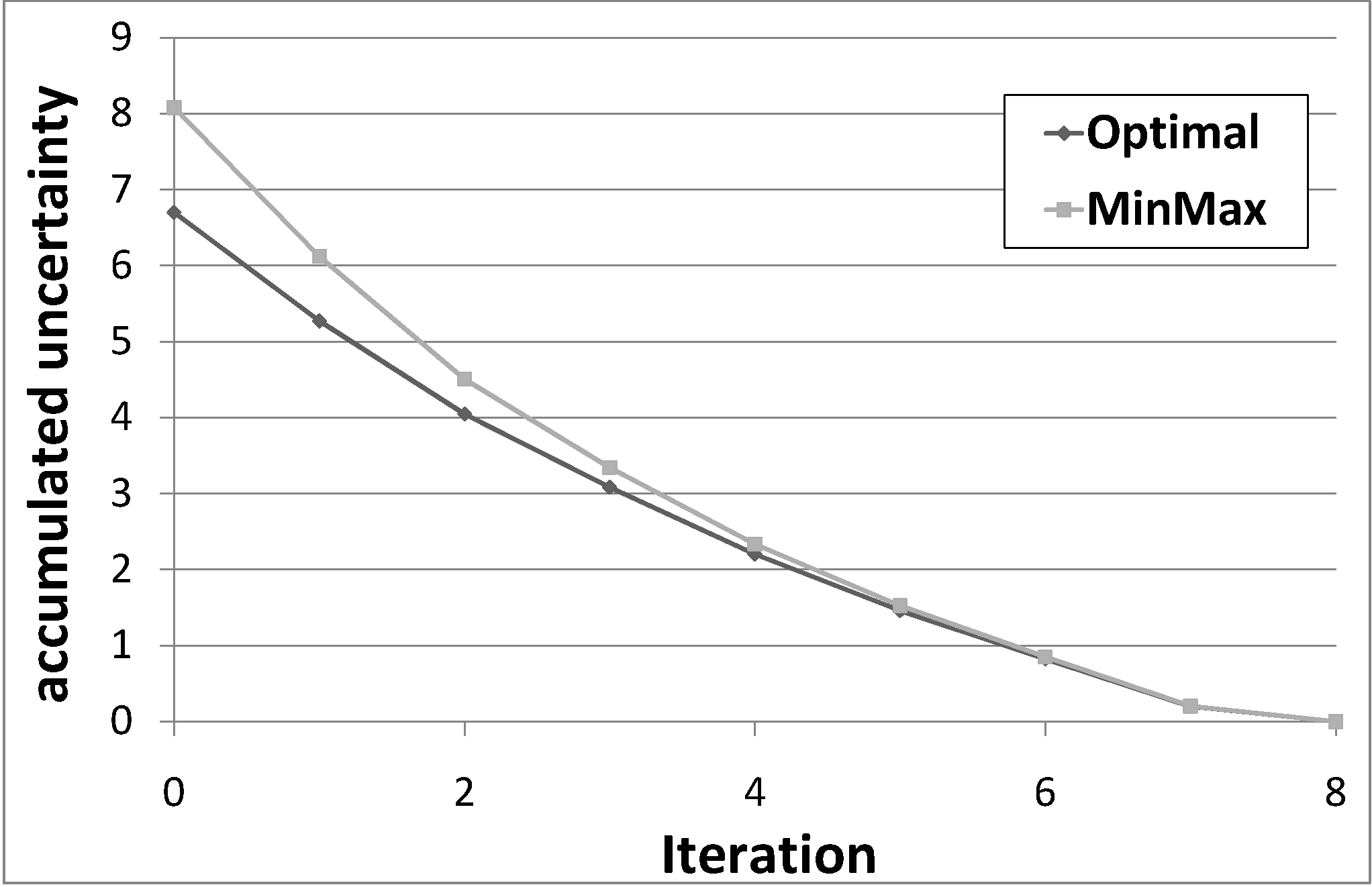}
    }
    \caption{Optimal vs. MinMax decision criterion.}
\end{figure}

\begin{figure}[t]
    \centering
    \subfigure[Synthetic Data]{
        \includegraphics[width = 0.4\columnwidth]{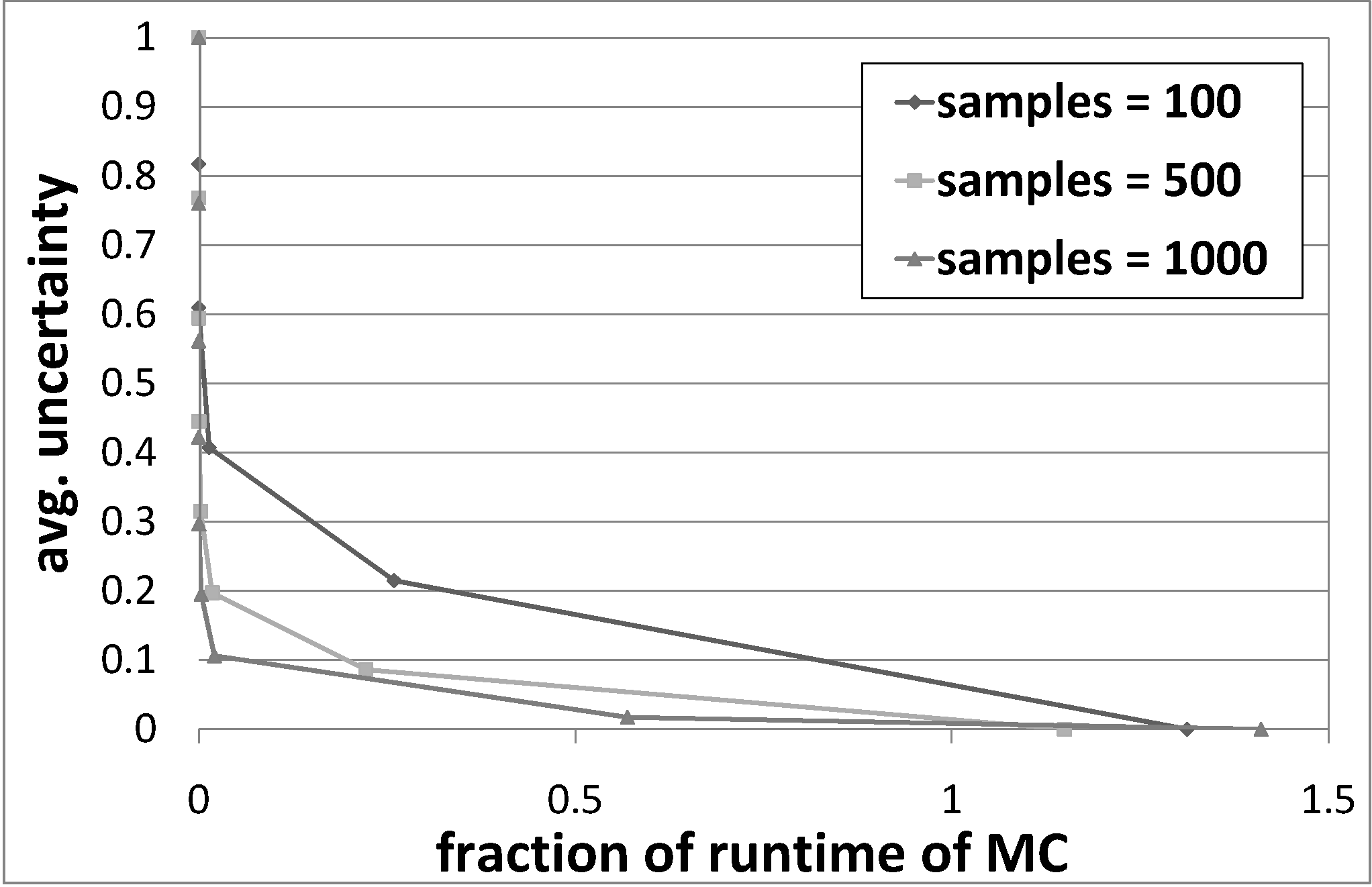}
    }
    \subfigure[Real Data]{
    \label{fig:real}
        \includegraphics[width = 0.4\columnwidth]{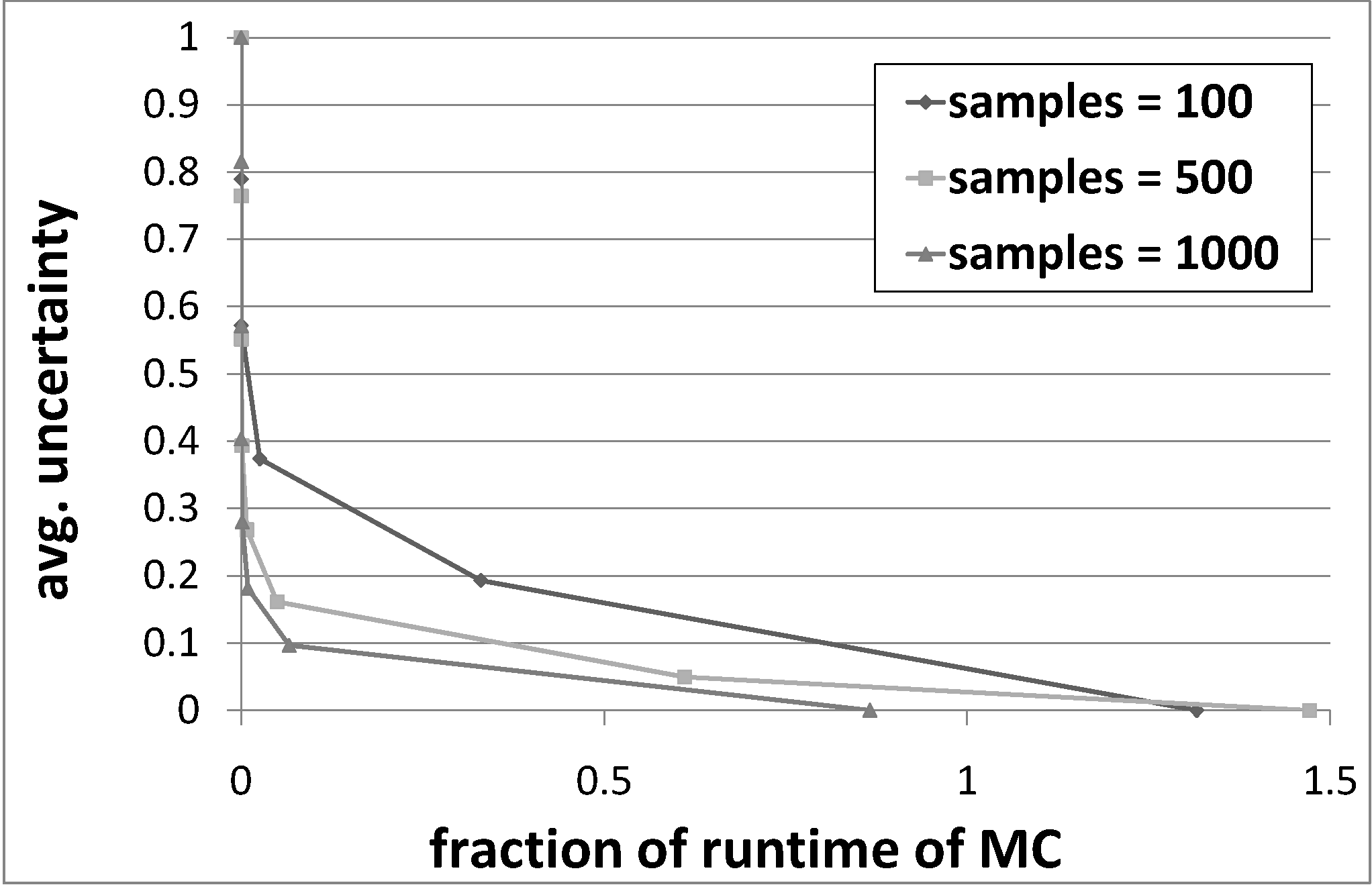}
    }
    \caption{{\small Uncertainty of \textbf{IDCA} w.r.t. the relative runtime
    to \textbf{MC}.}}

    \label{fig:iterative}
\end{figure}

\subsection{Optimal vs. Min/Max Decision Criterion}
In the first experiment, we evaluate the gain of pruning power
using the complete similarity domination technique (cf. Section
\ref{sub:spatialDomination}) instead of the state-of-the-art
min/max decision criterion to prune uncertain objects from the
search space. The first experiment evaluates the number of
uncertain objects that cannot be pruned using complete domination
only, that is the number of candidates are to evaluate in our
algorithm. Figure \ref{fig:cands} shows that our domination
criterion (in the following denoted as optimal) is able to prune
about 20\% more candidates than the min/max pruning criterion. In
addition, we evaluated the domination count approximation quality
(in the remainder denoted as uncertainty) after each decomposition
iteration of the algorithm, which is defined as the sum
$\sum_{i=0}^{N}DomCount_{UB}^i(B,R)-DomCount_{LB}^i(B,R)$. The
result is shown in Figure \ref{fig:optimal}. The improvement of
the complete domination (denoted as iteration $0$) can also be
observed in further iterations. After enough iterations, the
uncertainty converges to zero for both approaches.

\subsection{Iterative Domination Count Approximation}
Next, we evaluate the trade-off of our approach regarding
approximation quality and the invested runtime of our domination
count approximation. The results can be seen in Figure
\ref{fig:iterative} for different sample sizes and datasets. It
can be seen that initially, i.e. in the first iterations, the
average approximation quality (avg. uncertainty of an
\emph{influenceObject}) decreases rapidly. The less uncertainty
left, the more computational power is required to reduce it any
further. Except for the last iteration (resulting in 0
uncertainty) each of the previous iterations is considerably
faster than \textbf{MC}. In some cases (see Figure~\ref{fig:real})
\textbf{IDCA} is even faster in computing the exact result.

\begin{figure}[t]
    \centering
    \includegraphics[width = 0.4\columnwidth]{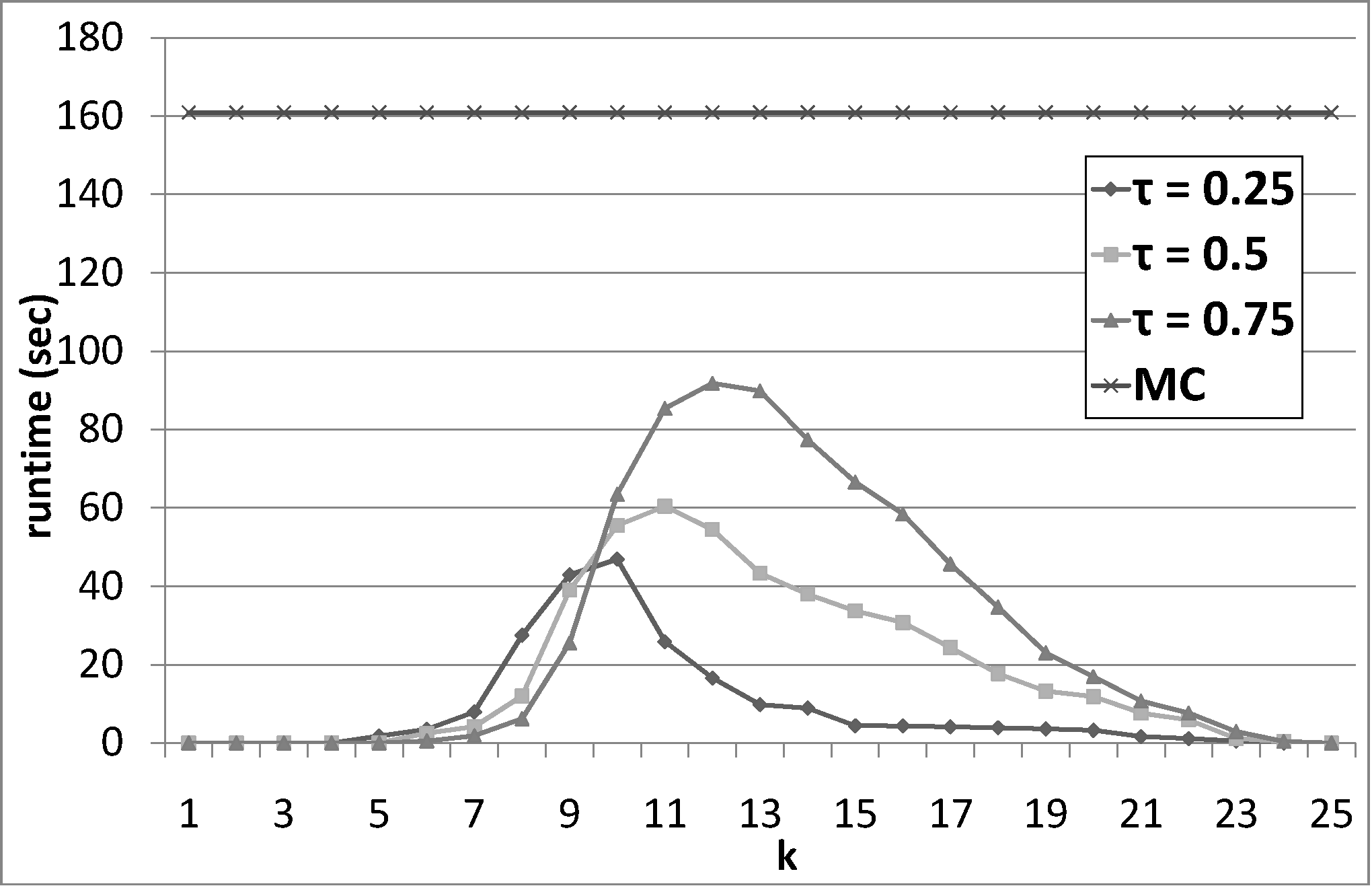}
    \caption{Runtimes of \textbf{IDCA} and \textbf{MC} for different
    query predicates $k$ and $\tau$.}
  \label{fig:predicate}
\end{figure}

\subsection{Queries with a Predicate}
Integrated in an application one often wants to decide whether an
object satisfies a predicate with a certain probability. In the
next experiment, we posed queries in the form: Is object $B$ among
the $k$ nearest neighbors of $Q$ (predicate) with a probability of
25\%, 50\%, 75\%? The results are shown in Figure
\ref{fig:predicate} for various $k$-values. With a given
predicate, \textbf{IDCA} is often able to terminate the iterative
refinement of the objects earlier in most of the cases, which
results in a runtime which is orders of magnitude below
\textbf{MC}. In average the runtime is below \textbf{MC} in all
settings.

\subsection{Number of influenceObjects}
The runtime of the algorithm is mainly dependent on the number of
objects which are responsible for the uncertainty of the rank of
$B$. The number of \emph{influenceObjects} depends on the number
of objects in the database, the extension of the objects and the
distance between $Q$ and $B$. The larger this distance, the higher
the number of \emph{influenceObjects}. For the experiments in
Figure \ref{fig:candidates-runtime} we varied the distance between
$Q$ and $B$ and measured the runtime for each iteration. In Figure
\ref{fig:scaling} we present runtimes for different sizes of the
database. The maximum extent of the objects was set to 0.002 and
the number of objects in the database was scaled from 20,000 to
100,000. Both experiments show that \textbf{IDCA} scales well with
the number influencing objects.

    \begin{figure}[t]
    \centering
    \subfigure[Runtime w.r.t. number of influence objects.]{
        \label{fig:candidates-runtime}
        \includegraphics[width = 0.4\columnwidth]{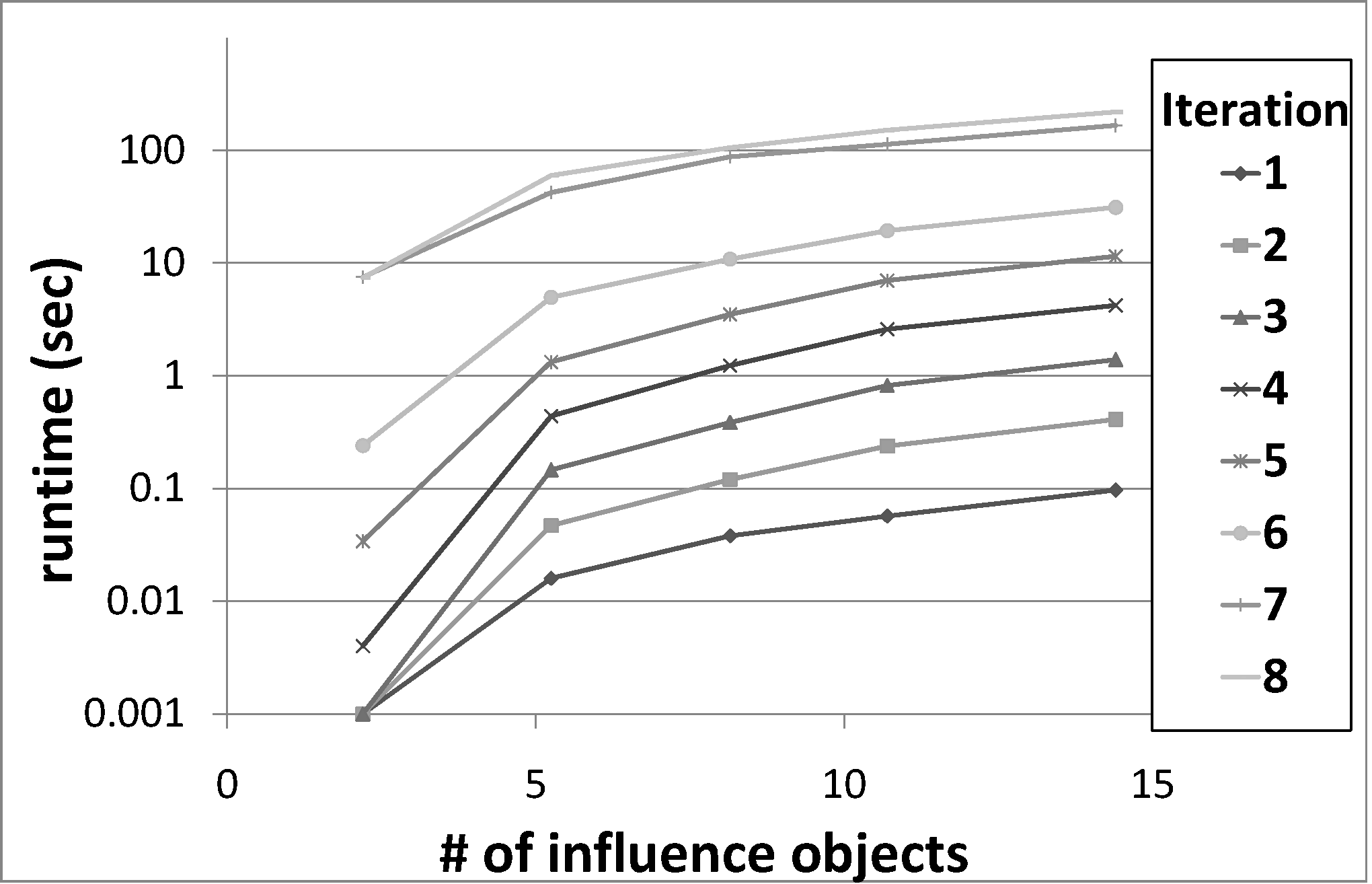}
    }
    \subfigure[Runtime for different sizes of the database.]{
        \label{fig:scaling}
        \includegraphics[width = 0.4\columnwidth]{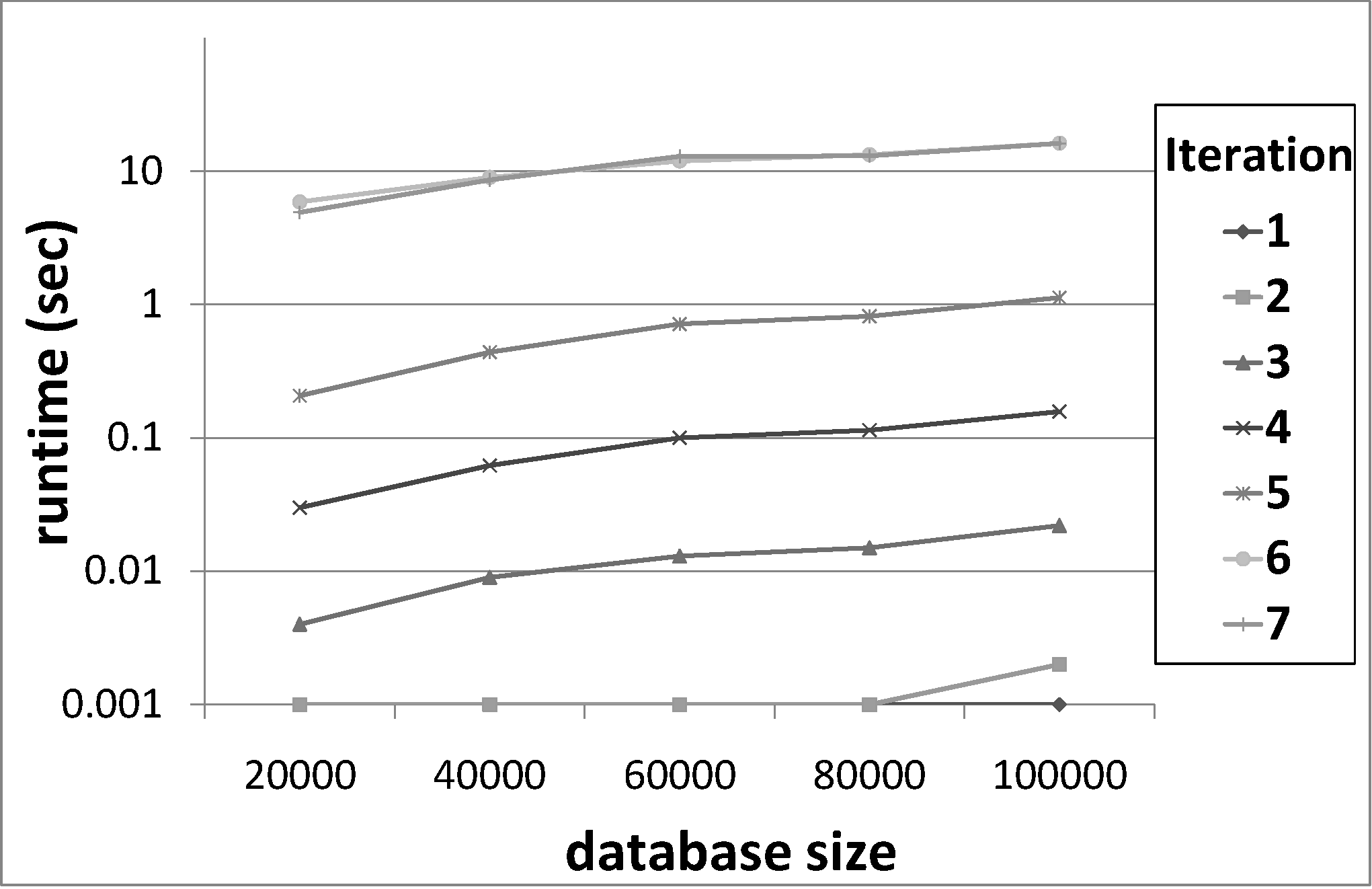}
    }
    \caption {Impact of influencing objects.}
\end{figure}

\section{Conclusions}
\label{sec:Conclusions} In this paper, we applied the concept of
probabilistic similarity domination on uncertain data. We
introduced a geometric pruning filter to conservatively and
progressively approximate the probability that an object is being
dominated by another object. An iterative filter-refinement
strategy is used to stepwise improve this approximation in an
efficient way. Specifically we propose a method to efficiently and
effectively approximate the domination count of an object using a
novel technique of uncertain generating functions. We show that
the proposed concepts can be used to efficiently answer a wide
range of probabilistic similarity queries while keeping
correctness according to the possible world semantics. Our
experiments show that our iterative filter-refinement strategy is
able to achieve a high level of precision at a low runtime. As
future work, we plan to investigate further heuristics for the
refinement process in each iteration of the algorithm. Furthermore
we will integrate our concepts into existing index supported
$k$NN- and R$k$NN-query algorithms.

\section*{Acknowledgements}
{This work was supported by a grant from the Germany/Hong Kong
Joint Research Scheme sponsored by the Research Grants Council of
Hong Kong (Reference No. G\_HK030/09) and the Germany Academic
Exchange Service of Germany (Proj. ID 50149322).

{
\bibliographystyle{abbrv}
\bibliography{abbrev,literature}
}

\end{document}